\documentclass[sigconf]{acmart}
\AtBeginDocument{%
  }

\setcopyright{acmlicensed}
\copyrightyear{2018}
\acmYear{2018}
\acmDOI{XXXXXXX.XXXXXXX}
\acmConference[Conference acronym 'XX]{Make sure to enter the correct
  conference title from your rights confirmation email}{June 03--05,
  2018}{Woodstock, NY}
\acmISBN{978-1-4503-XXXX-X/2018/06}

\usepackage{lipsum}
\usepackage{tikz}
\usepackage{amsmath}
\usepackage{stfloats}
\usepackage{url}
\usepackage{verbatim}
\usepackage{graphicx}
\usepackage{dsfont}
\usepackage{amsfonts}
\usepackage{array}
\usepackage{hyperref}
\usepackage[caption=false,font=normalsize,labelfont=sf,textfont=sf]{subfig}
\usepackage{algorithm2e}
\usepackage{footmisc}
\RestyleAlgo{ruled}
\SetKwComment{Comment}{/* }{ */}
\usepackage{booktabs}
\usepackage[table]{xcolor}
\usepackage{colortbl}
\usepackage{filecontents}
\usepackage{multirow}

\usepackage{amssymb}
\usepackage{amsthm}
\usepackage{mathtools}
\usepackage{bbm}
\usepackage{bm}
\usepackage{dsfont}
\usepackage{enumitem}
\usepackage{makecell}
\usepackage{slashbox}

\definecolor{Gray}{gray}{0.9}
\definecolor{LightCyan}{rgb}{0.88,1,1}

\newcolumntype{a}{>{\columncolor{Gray}}c}
\newcolumntype{b}{>{\columncolor{white}}c}

\newcommand{\mcl}[1]{\mathcal{#1}}

\newcommand{\mbb}[1]{\mathbb{#1}}

\newcommand{\prg}[1]{\vskip 6pt \noindent \textbf{#1}.}
\newcommand{\prginit}[1]{\noindent \textbf{#1}.}

\newcommand{\Lpix}[0]{L_{\text{pix}}}
\newcommand{\Lpct}[0]{L_{\text{pct}}}
\newcommand{\Lfrq}[0]{L_{\text{frq}}}
\newcommand{\Lsmt}[0]{L_{\text{smt}}}
\newcommand{\clip}[0]{\text{CLIP}}

\newcommand{\alex}[0]{\text{Alex}}
\newcommand{\onehot}[0]{\text{OneHot}}
\newcommand{\ppred}[0]{\hat{p}}
\newcommand{\R}[0]{\mbb R}

\newcommand{\one}[0]{\mathds{1}}
\newcommand{\enc}[0]{\text{enc}}

\newtheorem{thm}{Theorem}[section]

\newtheorem{lem}{Lemma}[section]

\numberwithin{equation}{section}

\begin{document}

\title{Hide\&Seek: Remove Image Watermarks with Negligible Cost via Pixel-wise Reconstruction}

\author{Huajie Chen$^1$, Tianqing Zhu$^{\clubsuit, 1}$, Hailin Yang$^1$, Yuchen Zhong$^1$, Yang Zhang$^2$, Hui Sun$^1$, Heng Xu$^1$, Zuobin Ying$^1$, Lihua Yin$^3$, Wanlei Zhou$^1$}

\renewcommand{\shortauthors}{Chen et al.}

\begin{abstract}
    Watermarking has emerged as a key defense against the misuse of machine-generated images (MGIs). 
    Yet the robustness of these protections remains underexplored. 
    To reveal the limits of state-of-the-art proactive image watermarking defenses, we propose HIDE\&SEEK (HS), a suite of versatile and cost-effective attacks that reliably remove embedded watermarks while preserving high visual fidelity. 
    Unlike prior approaches, HS only focuses on the pixels that are most critical to the watermark. 
    HS comprises two stages, namely HIDE and SEEK.
    In the HIDE stage, HS identifies and masks the pixels most critical to the watermark’s structure;
    in the SEEK stage, HS subsequently pixel-wise reconstructs only those with a generative model, leaving the remainder of the image untouched.
    This targeted strategy makes HS both \emph{query-free} and \emph{knowledge-free}, and does not requires access to the watermark detector nor knowledge of the underlying watermarking scheme. 
    Extensive experiments across multiple state-of-the-art watermarking methods show that HS consistently defeats current defenses and achieves better performance over existing removal attacks, whilst delivering stronger watermark erasure and higher image quality. 
    Our findings highlight fundamental vulnerabilities in current watermarking approaches and call into question their robustness as safeguards for MGIs.
    Our code will be released after acceptance.
\end{abstract}

\ccsdesc[500]{Security and privacy}
\ccsdesc[500]{Security and privacy~Digital rights management}

\keywords{AIGC Watermarking, Security Attack, Deep Learning}


\maketitle

\section{Introduction}

\begin{figure}[t!]
\centering
\includegraphics[width=\linewidth]{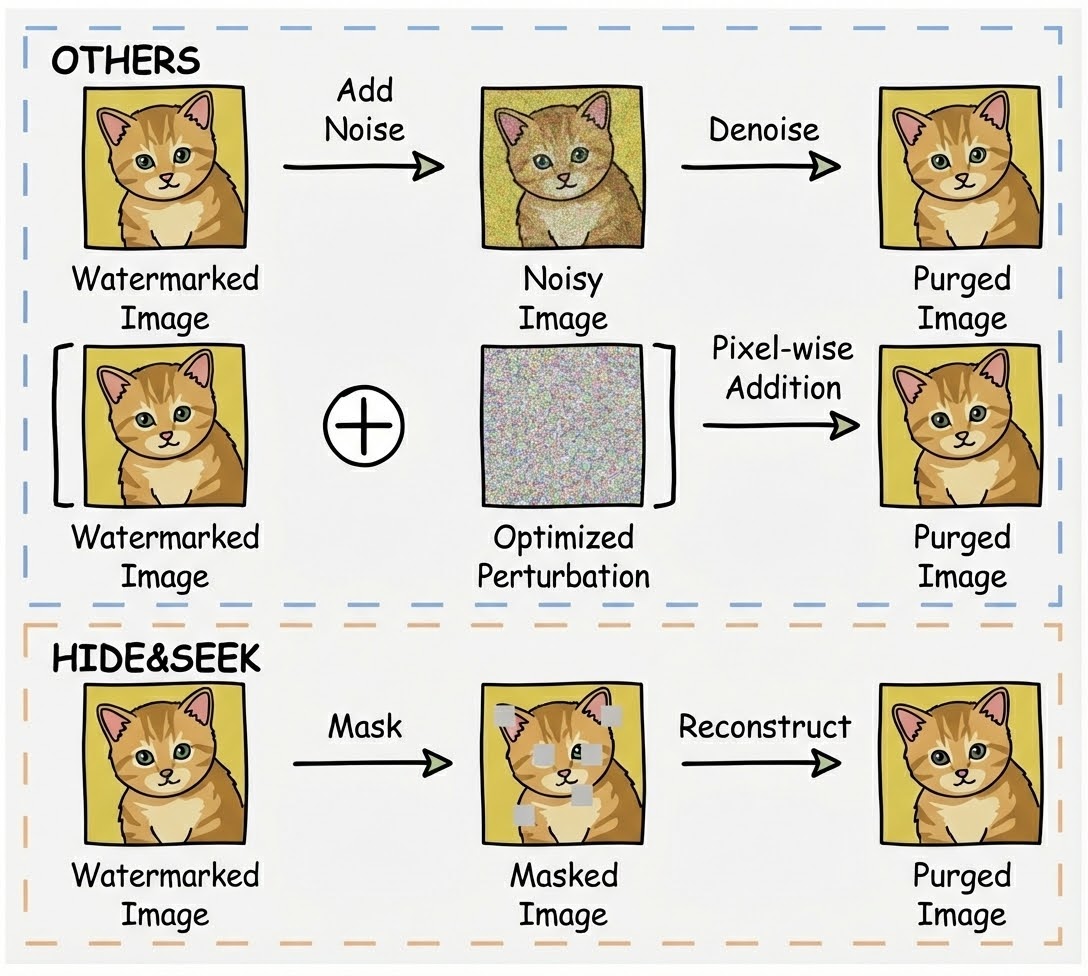}
\caption{\textbf{Overview of HIDE\&SEEK}. Compared to existing attacks that introduce global modification to the watermarked image so as to remove the embedded watermark, HIDE\&SEEK locates the vulnerable pixels and modifies them to purge the watermark while preserving high visual quality.}
\label{fig:overview}
\end{figure}

Machine-generated images (MGIs) has gained widespread adoption in many domains such as media~\cite{zhu2023genimage}, design~\cite{tang2025towards}, and entertainment~\cite{cao2025survey}. 
Despite their popularity and creative potential, MGIs have also been increasingly misused for malicious purposes, including political propaganda~\cite{shen2025gptracker}, explicit content~\cite{ma2025meme}, and hate speech~\cite{shen2025hatebench}. 
This prompts the development of many defense mechanisms to ensure the responsible use and distribution of MGIs~\cite{lukas2023ptw}.

Broadly, these defenses fall into two categories: passive MGI detection and proactive MGI watermarking. 
Passive MGI detection typically trains deep neural networks to distinguish MGIs from real images~\cite{zhao2021multi,frank2020leveraging}. 
Although such methods can achieve strong in-distribution performance, they often struggle to generalize to unseen generative or diffusion models~\cite{liu2023making}. 
Proactive watermarking, on the other hand, embeds imperceptible, verifiable signals (i.e., watermarks) into generated images at creation time, enabling post-hoc attribution and accountability~\cite{wen2023tree,zhu2018hidden,fernandez2023stable,yu2022responsible,yu2021artificial,lukas2023ptw,tancik2020stegastamp,jois2024pulsar}.
Owing to their model-agnostic guarantees, watermarking has emerged as a preferred approach for researchers and industry.

Nevertheless, a cat-and-mouse dynamic persists between watermark researchers and adversaries.
Watermarking remain vulnerable to removal strategies, such as regeneration~\cite{an2024waves,zhao2024invisible}, adversarial~\cite{jiang2023evading,saberi2024robustness}, and perturbation attacks~\cite{kassis2025unmarker}, that aim to erase embedded signals while preserving image fidelity. 
However, existing attacks face a fundamental trade-off among attack effectiveness, image fidelity, and computational efficiency.
To day, none achieves all three simultaneously to meaningfully undermine watermarking schemes in realistic settings.
This motivates us to address this gap. 
In this paper, we introduces a class of practical watermarking removal mechanisms that can maintain high attack success rates, preserve visual fidelity and, at the same time, achieve computationally efficient.
Our findings challenge the robustness of current watermarking schemes under realistic threat models and provide new insights into the evolving security landscape of MGI watermarking.

\prg{Our Work}
To reveal the limits of the SOTA watermarking defenses, we propose \textbf{Hide\&Seek}, an attack that removes image watermark via pixel-wise reconstruction.
In general, attacks on MGI watermarking schemes share the same objective ad similar approaches:
\textit{To create a purged image with the slightest modification added to the watermarked image}.
Furthermore, \textit{the modification must alter both high- and low-frequency components of the watermarked image so as to sabotage non-semantic and semantic watermark} \cite{kassis2025unmarker}.
Hence, our method is composed of two stages, namely ``Hide" and ``Seek", which respectively maximizes discrepancy of the watermarked and the purged image in frequency and semantic domain and minimizes the visual loss between them.

We first start with a simple version of the method that is named HIDE\&SEEK Naive (HSN).
The core idea of HSN is to modify certain regions of the watermarked image to create discrepancies in both the frequency and semantic domains and to generate the purged image.
To achieve this goal, the watermarked image is divided into multiple patches and masked with a selected strategy and a ratio in stage \textit{HIDE}.
Next, the masked patches are reconstructed using an auto-regressive model to create discrepancies in stage \textit{SEEK}.
By adjusting the masking strategy and the masking ratio, the adversary can modify only a part of the watermarked image, leaving the rest of the watermarked image untouched.
Thus, HSN creates a purged image with the least visual loss.

Even though HSN provides sufficient attack performance in evaluation, its masking and reconstruction schemes can still be improved.
For example, pixels at the edge of an object have greater impact on the high-frequency component of an image.
Therefore, if the adversary can measure the impact of pixels and apply different perturbation to the pixels based on the rank of the impact, the purged image will have better fidelity with guaranteed attack effectiveness.

We then further develop HSN into an elevated one named HIDE\&SEEK Plus (HS+).
In HS+, the adversary aims to locate pixels that have the most significant impact on the frequency and semantic domain of the watermarked image (denoted as vulnerable pixels).
These pixels are ranked with their vulnerability and reconstructed pixel-wise following a descending order based on the impact rank to incur the maximized discrepancies, where the most vulnerable pixel is reconstructed at last.
In stage \textit{HIDE}, the adversary uses a pre-trained masking model that takes an arbitrary watermarked image and produces a predicted mask.
The predicted mask carries the information about the vulnerability of pixels that guides both the masking and the reconstruction process.
The vulnerable pixels are then masked for reconstruction.
In stage \textit{SEEK}, the adversary utilizes a trained pixel generator that reconstructs a purged image from the masked image. 
The purged image is reconstructed pixel by pixel, where the most vulnerable pixel is reconstructed at last.
The idea is to utilize the intrinsic property of the auto-regressive model, where its previous output becomes its input in the next iteration.
The bias of the prediction therefore accumulates during the iterative reconstruction process.
Meanwhile, the generator is solely trained on clean images, which guarantees the visual and perceptual quality of the purged image.

Several research questions are encountered during the design and implementation process of Hide\&Seek:
\textbf{RQ1)}
How to create a mask that results in a sufficient discrepancy in both frequency and semantic domain with the least masked pixels?
\textbf{RQ2)}
How to ensure that the purged image has high fidelity while not preserving the embedded watermark?
\textbf{RQ3)}
How to guaranty attack performance with reduced computational cost?

For these questions, we have the following answers:
\textbf{A1)}
In HSN, a valid mask can be efficiently created with a selected strategy and a reasonable mask ratio. By adjusting the strategy and the ratio, the adversary controls the range and intensity of the modification. In HS+, a mask model is trained with specially designed loss functions to predict an optimal mask that covers the vulnerable pixels.
\textbf{A2)}
An auto-regressive model is trained in HSN to maintain high image fidelity. In HS+, a pixel-wise generator is trained with pixel loss and perceptual loss to recursively predict the masked pixels. Both generators are trained on clean image datasets to ensure that the purged image is watermark-free and visually similar to the watermarked image.
\textbf{A3)}
The masking model and generator can be efficiently trained and can be well generalized. Hence, no optimization is needed for an attack on a specific image.

HIDE\&SEEK is evaluated with comprehensive experiments.
Compared with other SOTA attack methods from top venues, Both HSN and HS+ show superior attack performance and visual quality along with moderate computational cost.
Additionally, HIDE\&SEEK reveals the vulnerability of current watermarking defenses.
This necessitates further research in this area to prevent AIGC misuses.
Our main contributions are listed as follows.
\begin{itemize}
    \item We propose HIDE\&SEEK, a versatile black-box, query-free attack that effectively removes watermarks in MGIs while maintaining high visual quality.
    \item Two versions of HIDE\&SEEK are devised to reveal the vulnerability of current SOTA watermarking defenses. HSN enables the adversary to launch efficient and effective attacks on the watermarked images by reconstructing the masked image with various masking strategies. 
    Further, HS+ creates a more delicate mask and employs a pixel-wise reconstruction technique to boost the attack performance with theoretical analysis.
    \item Comprehensive experiments are conducted to demonstrate the superior performance of HIDE\&SEEK compared with the SOTA baselines.
\end{itemize}

\section{Preliminaries \& Related Works}

\subsection{Defense for Machine Generated Images}

\prginit{Passive Detection}
Neural network detectors are trained to distinguish MGIs from real images based on high-frequency artifacts and semantic factitiousness \cite{li2020face,frank2020leveraging,zhao2021multi,liu2023making}.
In short, these methods aim to train a detector $\mcl D$ such that
\begin{equation}
\begin{aligned}
    \mathds{1}_{\mbb R^+}(\mcl D(X)) = \begin{cases}
        0 \quad \text{if $X$ is real}\\
        1 \quad \text{if $X$ is Fake} 
    \end{cases},
\end{aligned}
\end{equation}
where $\mathds{1}$ is the indicator function.
These approaches offer promising defensive results when they face the MGIs created by traditional generators, but fail to do so when advanced generators emerge \cite{kassis2023breaking}.
That is, if an image is produced by an unseen model, these approaches do not generalize well.

\prg{Proactive Watermarking}
Based on the taxonomy of Kassis et al. \cite{kassis2025unmarker}, current watermarking techniques for MGIs can be mainly divided into two classes: non-semantic and semantic watermarking.
These watermarking methods embed watermarks in the global spectral amplitudes associated with different levels of frequencies.
In general, these methods need to find a pair of embedding and retrieval functions $\mcl E(\cdot)$ and $\mcl R(\cdot)$ such that for any MGI $X$,
\begin{equation}
\begin{aligned}
    \min \mbb E \bigg[l(\mcl R(\mcl E(X, m) + \delta), m)\bigg],
\end{aligned}
\end{equation}
where $l$ denotes some distance metric function;
$\delta$ is some perturbative noise;
$m$ represents the bit-string watermark.

\textit{Non-semantic watermarking} methods \cite{zhu2018hidden,yu2021artificial,yu2022responsible,lukas2023ptw,fernandez2023stable,wang2024must,yang2024gaussian} embed the watermark into the high-frequency components of images.
These methods do not aim to change the content of images, and thus the watermark is embedded in the high-frequency components to which human eyes are less sensitive.
\textit{Semantic watermarking} methods \cite{tancik2020stegastamp,wen2023tree} embed watermarks in the low-frequency components of the images, which will drastically change the content of the images.
The reason is that the energy distribution is mainly concentrated in the low-frequency components.
Therefore, the robustness is enhanced with reduced visual quality.

\begin{table}[t!]
\centering
\label{tab:notations}
\caption{Notations}
\begin{tabular}{c|l}
\toprule
    Symbols & Definitions\\
\midrule
    $X, \tilde{X}$ & The original/reconstructed image\\
    \rowcolor{Gray}
    $\Lpix$ & The pixel loss\\
    $\Lfrq$ & The frequency loss\\
    \rowcolor{Gray}
    $\Lpct$ & The perceptual loss\\
    $\Lsmt$ & The semantic loss\\
    \rowcolor{Gray}
    $C, c$ & The number/index of channels\\
    $w, h$ & The width/height of the image\\
    \rowcolor{Gray}
    $M$ & The mask\\
    $I, \bar{I}, \hat{I}$ & The watermarked/masked/purged image\\
    \rowcolor{Gray}
    $E, D$ & The encoder/decoder model\\
    $H, G$ & The masking model/pixel generator\\
    \rowcolor{Gray}
    $(x, y)$ & The spatial coordinate of an image pixel\\
    $(u, v)$ & The frequency coordinate of an image pixel\\
    \rowcolor{Gray}
    $F$ & Two-dimensional discrete Fourier transformation\\
\bottomrule
\end{tabular}
\end{table}

\subsection{Detection Evasion Attacks}
\prginit{Regeneration Attacks}
In general, regeneration attacks first globally introduce noise $\delta$ to the watermarked image $X$ and then reconstruct the purged image $\tilde{X}$ from the noisy image, which is similar to the purification of adversarial examples \cite{nie2022diffusion,lee2023robust,tramer2019adversarial}.
This is defined as
\begin{equation}
\begin{aligned}
    \tilde{X} = G(X + \delta),
\end{aligned}
\end{equation}
where $G$ denotes the generator.
These attacks \cite{zhao2024invisible,zhou2024dataelixir,shi2023black} are \textit{black-box} and \textit{query-free}, meaning that the adversary does not need to have access to the detector.

\prg{Perturbation Attacks}
By optimizing a perturbative noise $\delta$ that is later introduced to the watermarked image $X$, the adversary removes the embedded watermark without the regeneration process, which is similar to adversarial perturbation \cite{hayes2018learning}.
This is defined as
\begin{equation}
\begin{aligned}
    \underset{\delta}{\arg \min} \ \mbb E [\|\delta\|^2 + \mcl D(X + \delta)].
\end{aligned}
\end{equation}
This category of attacks \cite{jiang2023evading,lukas2024leveraging,saberi2024robustness,kassis2025unmarker} applies global perturbations to watermarked image to remove the embedded watermark.

\subsection{Mask Autoencoders}
He et al. \cite{he2022masked} first propose Mask AutoEncoders (MAEs) as scalable vision learners. MAEs adopt the concept from auto-regressive language modeling and masked auto-encoding, where a portion of the data is masked, and MAEs learn to predict the masked content.
Li et al. \cite{li2025fractal} further devise the Fractal Generative Model (FGM) that recursively generates images pixel-wise.
During the recursion process, the next predicted pixel is created with the global state of the current masked image.
Notably, compared to DDPMs or Generative Adversarial Networks (GANs), training a FGM is relatively easier, because it does not have the model collapse phenomenon like GANs or requires multiple steps in the reverse process like DDPMs.
A cross-entropy loss function is employed to guide the training of FGM such that it predicts the RGB channels of a pixel as tokens in LLMs.

\section{Problem Formulation}

\subsection{Threat Model}
We consider two parties in this attack scenario: \textbf{Victim} and \textbf{Adversary}.
A \textbf{victim} is an organization or a company that establishes the GenAI platform and provides API services where the victim monetizes.
An \textbf{adversary} is a malicious user who attempts to evade the defense mechanisms and create harmful MGIs.

\prg{Victim's Goal}
The victim wants to ensure that the GenAI platform is not utilized to create malicious AIGC that spreads misinformation.
All images from the GenAI platform are therefore watermarked for MGI detection.
Given an arbitrary suspect image, the victim feeds the image into the detector network to extract the embedded watermark that determines whether it is a MGI.

\prg{Victim's Capabilities}
The victim is allowed to access all existing model structures and training algorithms. 
Adequate resources such as datasets and GPUs are granted to the victim such that GenAI models can be trained.
Additionally, GenAI APIs receive queries and return machine-generated images.
Intermediary results are not required to be given to users.

\prg{Adversary's Goal}
The adversary aims to utilize the victim's GenAI APIs to generate high-quality MGIs while removing the embedded watermark inside such that the image will not be determined as a MGI by any existing algorithms.

\prg{Adversary's Capabilities}
Due to insufficient resources, the adversary cannot afford to purchase high-quality datasets or rent servers for long-term large model training.
Direct access to the training dataset and the model structure used by the victim is not allowed for the adversary.
The adversary cannot interact with the detector of the victim.
However, it is possible for the adversary to collect an auxiliary dataset that shares a distribution similar to the training dataset of the victim.
The adversary can also access limited computational resources to train relatively simple local models.

\subsection{Formulated Objective}
\label{sect:form_obj}

\prginit{Attack Properties}
We may further conclude that a good watermark removal attack must possess the following properties:
\begin{itemize}
\item \textbf{Effectiveness}: The attack must ensure that no watermark can be detected or extracted from the purged image.
\item \textbf{Fidelity}: The purged image must not deviate too much from the watermarked image.
\item \textbf{Efficiency}: The resources and time consumption of the attack should be low.
\item \textbf{Versatility}: The attack can be applied to any image that is watermarked with any watermarking method.
\end{itemize}

\prg{Objectives}
The aforementioned attacks either introduce noises to the watermarked image to directly purge the embedded watermark or reconstruct the purged image from the noisy image.
Abstractly, these approaches can be concluded into one sentence:
\textit{The adversary removes the watermark by selectively adjusting the pixel values of the watermarked image.}
Thus, which pixels are to be changed and how much the pixels should shift together form the critical problem.

Let $I \in \R^{C \times w \times h}$ denote a watermarked image generated by the victim's GenAI APIs, where $C, w, h$ respectively denote the number of channels, width, and height of $I$.
The victim owns a detector $D: \R^{C\times w\times h} \mapsto \{0, 1\}$ that is able to determine whether a given image contains a watermark or not.
The ultimate goal of the adversary is to create a watermark-free image $\hat{I}$, also known as a purged image, such that $D(\hat{I})=0$, indicating that there is no watermark in $\hat{I}$.
In addition, $\hat{I}$ is required to be as similar as possible to $I$.
We then formulate this goal into an optimization problem:
\begin{equation}
\begin{aligned}
    \underset{\hat{I}}{\min} \ \mbb E \big[ \big \|\hat{I} - I \big\|^2 -\log \big(1- D(\hat{I})\big) \big].
\end{aligned}
\end{equation}
However, accessing $D$ is prohibited for the adversary in our settings.
Other metrics for image quality measurement must also be considered.
Hence, the optimization is further specified as follows.

It is demonstrated that robust watermarks reside in the global spectral amplitude \cite{kassis2025unmarker}.
This results in two schemes: semantic and non-semantic watermarks.
On the one hand, non-semantic watermarks refer to those watermarks hidden within the global spectral amplitudes of the images where the high-frequency components are concentrating.
The high-frequency components are less perceptible to human eyes, and are therefore intrinsically selected by deep-learning-based watermarking schemes.
However, such watermarks can be easily removed with slight perturbations on the high-frequency components \cite{chen2024high}.
On the other hand, semantic watermarks are embedded into the low-frequency components of the images that significantly influence the image perceptually.
Hence, such watermarks are inevitably visible to human eyes, but are more robust to perturbations.

Based on this assumption, the attack objective can be further formulated in the following form:
\begin{equation}
\begin{aligned}
    \underset{\hat{I}}{\max} \ \mbb E \big[ d_{\text{frq}}(\hat{I}, I) + d_{\text{smt}}(\hat{I}, I)\big],
\end{aligned}
\end{equation}
where $d_{\text{frq}}(\cdot, \cdot)$ and $d_{\text{smt}}(\cdot, \cdot)$ respectively denote the distance measure of frequency and semantic domain of the two input images.
The adversary thereby maximizes the frequency and the semantic discrepancy between $\hat{I}$ and $I$ so as to purge the watermark. 

Meanwhile, the adversary does not want to sacrifice too much image quality to remove the watermark.
Thus, the adversary aims to preserve both the pixel and the perceptual similarity between $\hat{I}$ and $I$, which is then formulated into
\begin{equation}
\begin{aligned}
    \underset{\hat{I}}{\min} \ \mbb E \big[ d_{\text{pix}}(\hat{I}, I) + d_{\text{pct}}(\hat{I}, I) \big],
\end{aligned}
\end{equation}
where $d_{\text{pix}}$ and $d_{\text{pct}}$ respectively denote the distance measure of pixels and perceptual similarity.
The adversary thereby minimizes the pixel and perceptual distances between the two input images so that the visual quality of $\hat{I}$ is maintained.

Eventually, we may combine all the above optimization objectives together to derive the below optimization problem:
\begin{equation}
\begin{aligned}
\underset{\hat{I}}{\min} \ \mbb E \big[ & \big(d_{\text{pix}}(\hat{I}, I) + d_{\text{pct}}(\hat{I}, I)\big) \\
 & - \big(d_{\text{frq}}(\hat{I}, I) + d_{\text{smt}}(\hat{I}, I)\big) \big].
\end{aligned}
\end{equation}
In general, the adversary aims to create the maximum discrepancy in the frequency and semantic domain of $\hat{I}$ and $I$ with the minimum visual loss.

\section{Hide\&Seek}

\subsection{Overview}
Despite the solid attack performance provided by previous studies, we discover that there are still several issues to address.
Regeneration attacks \cite{jiang2023evading} are highly effective in dealing with non-semantic watermarking methods, but fail to offer the same results when faced with semantic watermarking methods.
Perturbation attacks \cite{kassis2025unmarker} show promising results for both semantic and non-semantic watermarking methods.
Nevertheless, visual quality loss of the purged image caused by the perturbation is non-negligible, because the perturbation must affect the low-frequency components of the purged image.
Additionally, for each watermarked image, the perturbative noise is specifically optimized, resulting in a high computational cost.
Ultimately, the above attacks produce the purged image by globally modifying the watermarked image, which intrinsically diminishes the visual quality of the purged image.

We thus devise our watermark removal attack HIDE\&SEEK that consists of two stages: \textit{HIDE} and \textit{SEEK}.
Overall, in stage \textit{HIDE}, the adversary attempts to create a mask to \textit{hide} the pixels given an arbitrary watermarked image;
in stage \textit{SEEK}, the adversary then tries to \textit{seek} the hidden pixel values.
With the two stages, the adversary completes the watermark removal attack by modifying the values of the masked pixels while leaving the rest of the pixels untouched.
Compared to existing attacks, HIDE\&SEEK distinguishes itself from them by removing watermarks through partial image modification.
We further develop two versions of the attack, namely, HIDE\&SEEK Naive (HSN) and HIDE\&SEEK Plus (HS+).

\subsection{HIDE\&SEEK Naive}

\begin{figure*}[t!]
\centering
\includegraphics[width=.9\textwidth]{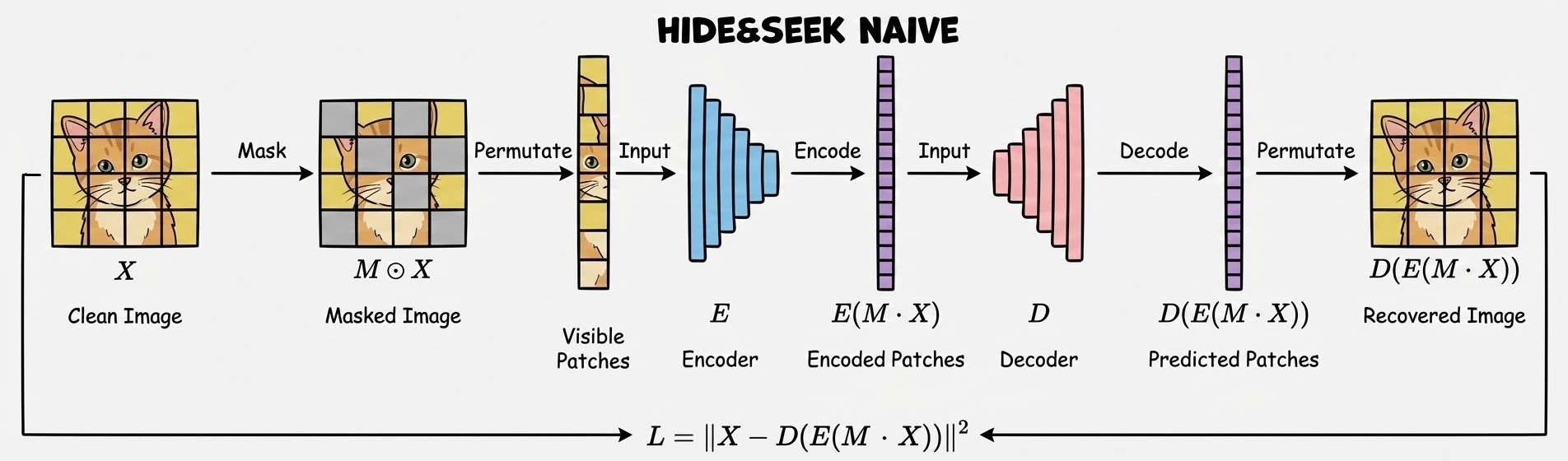}
\caption{\textbf{HIDE\&SEEK Naive Workflow}. The adversary first masks the watermarked image with a mask using a strategy and a ratio. Those visible patches are then encoded by the encoder and concatenated with the masked patches, which are then decoded by the decoder to derive the purged image.}
\label{fig:hsn_workflow}
\end{figure*}

\begin{figure}[t!]
\centering
\includegraphics[width=.9\linewidth]{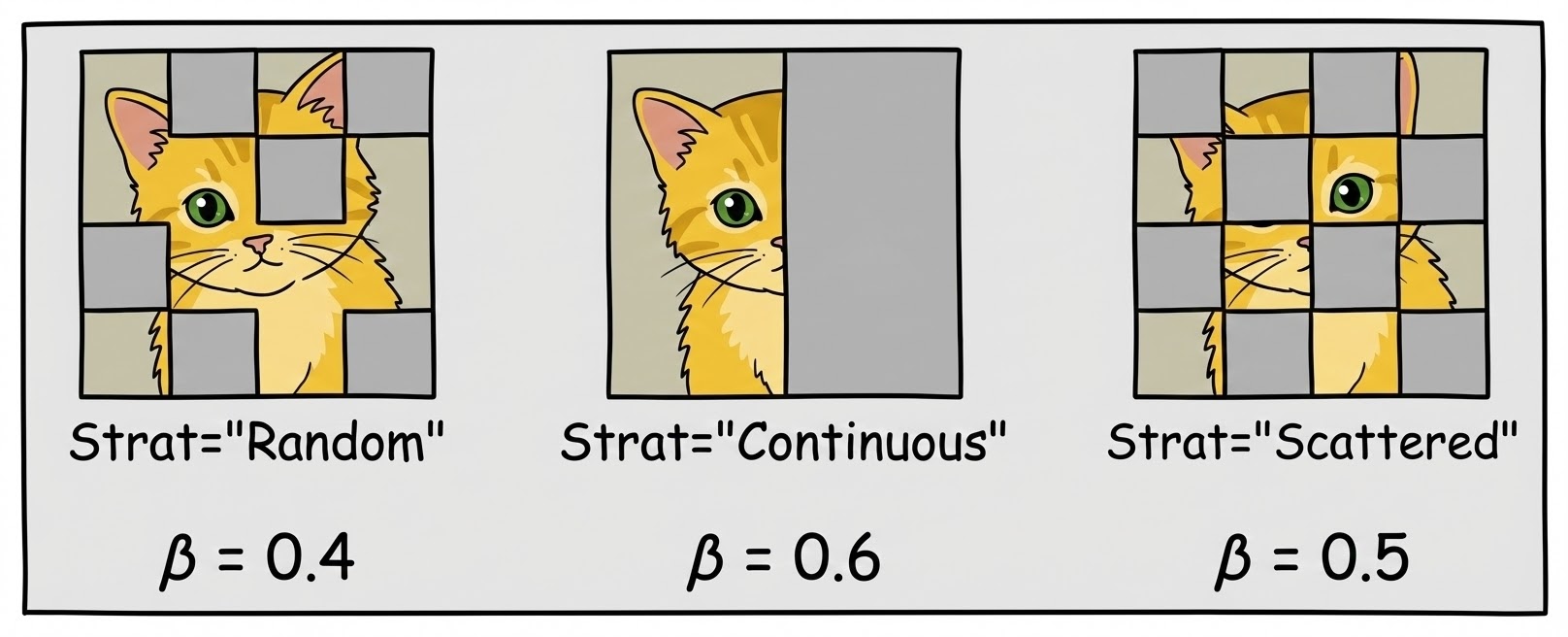}
\caption{\textbf{Masking Strategy}. HSN offers three types of masking strategies: Random, Continuous, and Scattered, where $\beta$ controls the masking ratio.}
\label{fig:masking_strategy}
\end{figure}

HSN is a simple yet effective version of HS.
The mask is generated by a fixed algorithm with a selected strategy and a hyperparameter $\gamma$ that controls the masking ratio.

\prg{Training}
As illustrated in \autoref{alg:hsn} and \autoref{fig:hsn_workflow}, to launch HSN, the adversary collects a clean dataset $\{X^1,...,X^n\}$ first, where $X^i \in \R^{C\times w\times h}$ denotes a watermark-free image.
Then, the adversary trains the generator model.
In HSN, the Masked Auto-Encoder (MAE) \cite{he2022masked} is employed as the architecture of the generator.
An MAE consists of an encoder $E:\R^{C \times w \times h} \mapsto \R^Q$ and a decoder $D:\R^Q \mapsto \R^{C \times w \times h}$, where $Q$ denotes the dimensionality of embeddings in the latent space.

During the training process, a mask $M \in \{0, 1\}^{w \times h}$ is randomly generated.
A clean image $X \in \{X^1, ..., X^n\}$ is then masked with $M$, denoted as $M \odot X$, where $\odot$ is the element-wise multiplication operator.
Next, $E$ takes in $M \odot X$ and encodes the visible pixels into tokens.
The masked pixels are tokenized and concatenated with the tokens together to form a vector $E(M \odot X)$ in the latent space following the order of pixels in $X$.
Eventually, $D$ takes in $E(M \odot X)$ to predict the masked pixels using the visible tokens to create the purged image $D(E(M \odot X))$.
By minimizing the regeneration loss between the clean image $X$ and the purged image $D(E(M \odot X))$, the adversary completes the training of $G$.
The regeneration loss is defined as
\begin{equation}
\begin{aligned}
L = \|X - D(E(M \odot X))\|^2,
\end{aligned}
\end{equation}
where only the masked pixels are included in the loss computation in order to lower the computational cost.

\prg{Masking Strategies}
Different types of masking eventually result in various effects on purged images.
As illustrated in \autoref{fig:masking_strategy}, we roughly divide masking strategies into three classes:
\textit{Random, Continuous, and Scattered}.
An image is first divided into multiple patches.
A random mask is spontaneously created;
a continuous mask is created by randomly selecting an initial patch.
The next patch to mask must be adjacent to the prior one;
a scattered mask must not have any patch that is adjacent to the other selected patches.

What distinguishes these strategies from each other is the effect brought by each of them \cite{he2022masked}.
The continuous mask has the most significant impact on the low-frequency components of $I$, whereas the scattered mask preserves the most frequency information.
The random mask allows greater $\beta$, meaning that more pixels can be masked with guaranteed reconstruction quality.
But its effect on the frequency domain of $I$ often fluctuates because of its randomness.

\prg{Evaluation}
As illustrated in \autoref{alg:hsn}, in evaluation, the adversary first receives a watermarked image $I \in \R^{C \times w \times h}$.
A mask $M \in \{0, 1\}^{w \times h}$ is then created with a masking strategy, and the masking ratio is controlled by the hyperparameter $\beta$.
The adversary eventually gets the purged image $\hat{I} = D(E(M \odot I))$ using the trained MAE.

\prg{Limits}
Even though HSN provides reliable watermark removal attack performance in evaluation, the area and the extent of the modification applied to the watermarked image can still be optimized.
The randomly generated mask with the selected strategy needs to cover a pretty high ratio of the entire area of the image.
Some of the masked regions may not contribute to watermark removal, and they are thus redundant to be modified.
Additionally, the MAE cannot provide the state-of-the-art image regeneration quality \cite{zhang2022mask}.
To this end, we need a more sophisticated masking model to locate the critical regions, and a more delicate generator to recover those regions.
Thus, we propose the advanced version of HIDE\&SEEK, namely HIDE\&SEEK Plus (HS+).

\begin{algorithm}[t!]
\caption{HIDE\&SEEK Random}\label{alg:hsn}
\tcc{Training}
\KwData{Training Dataset $\{X^1,...,X^n\}$}
\KwResult{Trained Encoder $E$, Decoder $D$}
\While{$L$ not converged}{
    $X \gets \text{sample}(\{X^1,...,X^n\})$\;
    $M \gets \text{create\_random\_mask}()$\;
    $L \gets \|X - D(E(M \odot X))\|^2$\Comment*[r]{Loss}
    Optimizer steps\;
}
\Return $E, D$\;

\vskip 5pt

\tcc{Evaluation}
\KwData{Watermarked Image $I$, Mask Proportion $\beta$}
\KwResult{Purged Image $\hat{I}$}
$M \gets \text{create\_mask\_by\_strategy}(\beta)$\;
$\hat{I} \gets D(E(M \odot I))$\;
\Return $\hat{I}$\;
\end{algorithm}

\subsection{HIDE\&SEEK Plus}

\begin{figure*}[t!]
\centering
\includegraphics[width=.9\textwidth]{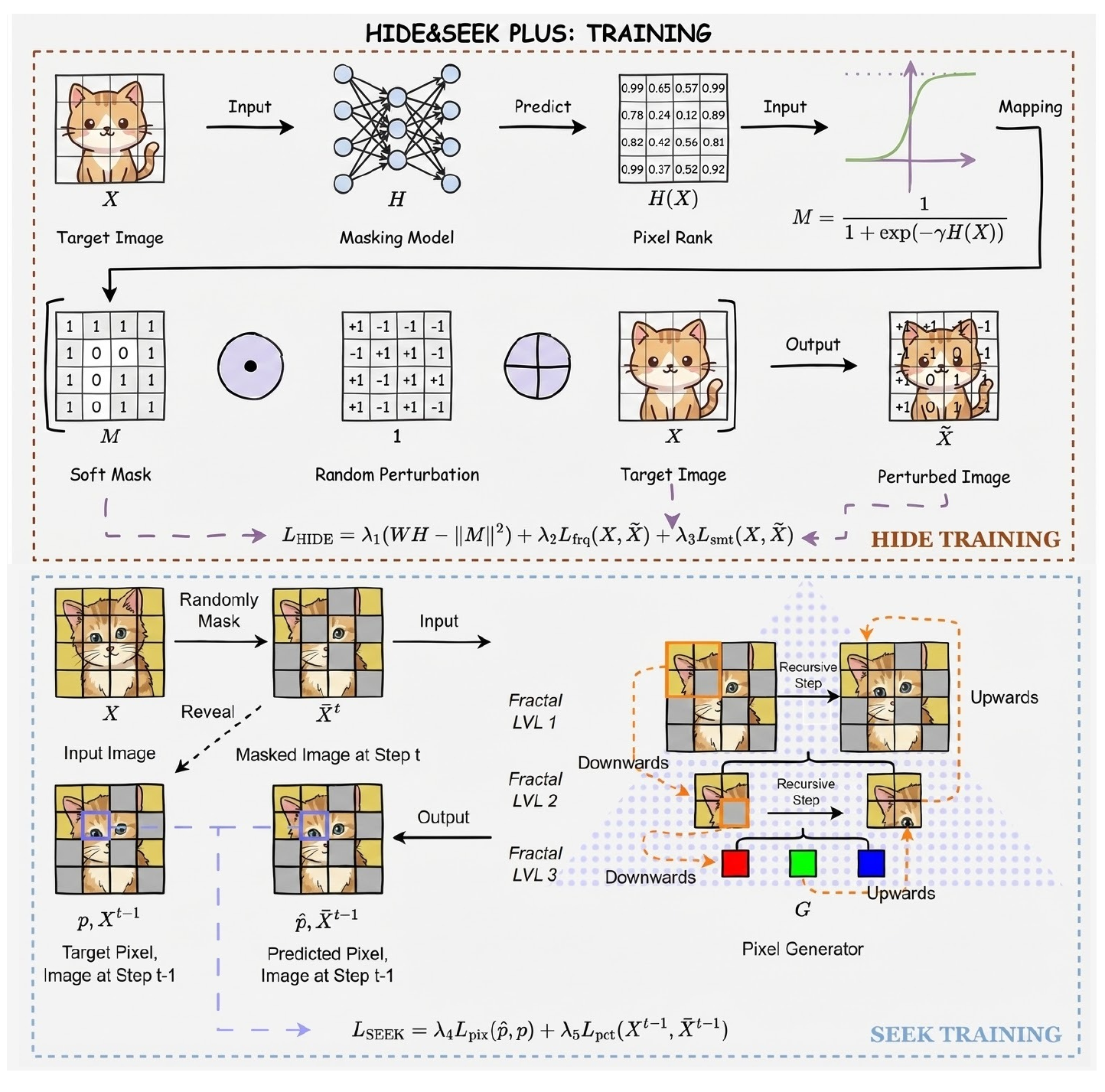}
\caption{\textbf{HIDE\&SEEK Plus Training Workflow}. In stage \textit{HIDE}, the masking model takes a target image to predict a rank for all pixels in the image. The rank is then mapped to a soft mask for random perturbation. The soft mask, the perturbed image and the target image are used to compute the HIDE loss. In stage \textit{SEEK}, A randomly masked image is fed into the pixel generator to predict the next pixel. The SEEK loss is then computed using the predicted pixel and the target pixel.}
\label{fig:hsp_workflow_training}
\end{figure*}

The core idea of HS+ is to locate the pixels that have the most significant impact on the watermarked image $I$ given the slightest modification, which is therefore named ``\textit{the vulnerable pixels}".
Then, these pixels are masked and reconstructed one by another following a ranked order, where the most vulnerable pixel is reconstructed at last in order to create larger bias.




\prg{Training}
To achieve this, the adversary trains a masking model $H: \R^{C \times w \times h} \mapsto R^{w \times h}$ that predicts the masking probability of pixels given a watermark-free image $X \in \{X^1,...,X^n\}$, as illustrated in \autoref{alg:hsp} and \autoref{fig:hsp_workflow_training}.
Given the logits $H(X)$, the adversary then needs to measure the vulnerability of the pixels so as to get the rank.
Since semantic and non-semantic watermarks reside in the high and low-frequency components of $I$, the selected pixels must have maximum impact on both the frequency and the semantic domain on $I$.
We first create a soft mask with a differentiable operation defined as follows:
\begin{equation}
\begin{aligned}
M = \frac{1}{1 + \exp[-\gamma(H(X))]},
\end{aligned}
\end{equation}
where $\gamma$ is a temperature parameter;
$\exp[\cdot]$ denotes the element-wise exponentiation.
For example, in $\exp[X]$, each element $X_{ij} \in X$ will be used as an exponent to create a matrix that shares the dimensionality with $X$.
Each element in $\exp[X]$ will become $e^{X_{ij}}$.
Thereby, $H(X)$ is scaled so that positive values are mapped close to $1$, whereas negative values are mapped close to $0$.
$M$ is thus called a soft mask because this scaling does not guarantee an exact mapping to $0$ or $1$.

The adversary uses the soft mask $M$ to give $X$ the minimal modification.
The least modification of a pixel should be $\pm1$, because a colorful pixel is an integer $p \in [0, 255]^3$.
Thus, the adversary randomly creates a perturbation $\one \in \{-1,1\}^{C \times w \times h}$.
$X$ is perturbed with $\one$ and $M$ to get the perturbed image $\tilde{X}$, denoted as $\tilde{X} = M \odot \one + X$.

The frequency discrepancy between $X$ and $\tilde{X}$ is measured by a frequency loss defined as 
\begin{equation}\label{eq:frq_loss}
\begin{aligned}
    \Lfrq(X, \tilde{X}) = & \frac{1}{wh} \sum_{u=0}^{w-1} \sum_{v=0}^{h-1} \omega(u, v) \cdot \\
    & |F_X(u, v) - F_{\tilde{X}}(u, v)|^2,
\end{aligned}
\end{equation}
where $\omega$ is the spectrum weight matrix defined as
\begin{equation}\label{eq:weight_spat_freq}
\begin{aligned}
    \omega(u,v) = & |F_X(u, v) - F_{\tilde{X}}(u, v)|^\alpha.
\end{aligned}
\end{equation}
Here, $\alpha$ is the scaling factor and is set to $1$ by default.
$F_X(u, v)$ denotes the complex frequency value of the 2D Fourier transform at the coordinate $(u, v)$ on the frequency spectrum, defined as
\begin{equation}\label{eq:2dfft}
\begin{aligned}
    F_X(u, v) & = \sum_{x=0}^{w-1} \sum_{y=0}^{h-1} X_{x, y} \exp(-i2\pi(\frac{ux}{w} + \frac{vy}{h}))\\
    & = a + bi,
\end{aligned}
\end{equation}
where $a$ and $bi$ respectively denote the real and imaginary part of the complex value $F_X(u,v)$.

The semantic discrepancy between $X$ and $\tilde{X}$ is measured by a semantic loss defined as
\begin{equation}\label{eq:smt_loss}
\begin{aligned}
    \Lsmt(X, \tilde{X}) = \| \clip(X) - \clip(\tilde{X}) \|^2,
\end{aligned}
\end{equation}
where CLIP \cite{radford2021clip} denotes the Contrastive Language-Image Pretraining model that maps input in different modality (text and image) into the same latent space.
Therefore, by measuring the Euclidean distance between the features of $X$ and $\tilde{X}$ in the latent space, we can measure the semantic discrepancy.


The goal of the adversary is to modify the least number of pixels to create the max discrepancies in the frequency and the semantic domains between the watermarked image and the purged image.
In the training process, we need to transform this problem into a loss minimization problem.
Hence, we can rewrite the goal, where the adversary modifies as many pixels as possible to create the minimum discrepancies in both the frequency domain and the semantic domain between the watermarked image and the purged image.
We can thus derive the following loss function:
\begin{equation}\label{eq:loss_hide}
\begin{aligned}
L_{\text{HIDE}} = & \lambda_1 (WH - \|M\|^2) + \lambda_2 \Lfrq(X, \tilde{X}) \\
& + \lambda_3 \Lsmt(X, \tilde{X}),
\end{aligned}
\end{equation}
where the $\lambda$s denote the weight parameters;
the first term measures the distance between the image area and the mask area.
By minimizing $L_{\text{HIDE}}$, $H$ gives lower scores to more vulnerable pixels.


\begin{algorithm}[t]
\caption{HIDE\&SEEK Plus}\label{alg:hsp}
\tcc{Training}
\KwData{Training Dataset $\{X^1,...,X^n\}$}
\KwResult{Trained Masker $H$ and Generator $G$}
\While{$L_{\text{HIDE}}$ not converged}{
    \tcc{Train $H$}
    $X \gets \text{sample}(\{X^1,...,X^n\})$\;
    $M \gets \frac{1}{1 + \exp[-\gamma H(X)]}$\Comment*[r]{Predict Soft Mask}
    $\tilde{X} \gets M \odot \one + X$\Comment*[r]{Get Perturbed Image}
    $L_{\text{HIDE}} \gets \lambda_1 \|M\|^2 + \lambda_2 \Lfrq(X, \tilde{X}) \newline + \lambda_3 \Lsmt(X, \tilde{X})$\Comment*[r]{Compute Loss}
    Optimizer steps\;
}
\While{$L_{\text{SEEK}}$ not converged}{
    \tcc{Train $G$}
    $X \gets \text{sample}(\{X^1,...,X^n\})$\;
    $M \gets \text{create\_random\_mask}()$\;
    $\bar{X}^{|M|} \gets M \odot X$\Comment*[r]{Mask Image}
    $\hat{p} \gets G(\tilde{X}^{|M|})$\Comment*[r]{Predict the next token}
    $\bar{X}^{|M|-1} \gets \enc(\bar{X}^{|M|}, \hat{p})$\;
    $L_{\text{SEEK}} \gets \lambda_4 \Lpix(\hat{p}, p) + \lambda_5 \Lpct(X^{|M|-1}, \bar{X}^{|M|-1})$\Comment*[r]{Compute Loss}
    Optimizer steps\;
}
\Return $H, G$\;

\tcc{Evaluation}
\KwData{Watermarked Image $I$}
\KwResult{Purged Image $\hat{I}$}
$M \gets \text{get\_hard\_mask}(\sigma(H(I)))$\Comment*[r]{Hard Mask}
$I^{|M|} \gets M \odot I$\Comment*[r]{Mask Image}
\For{$t \in [|M|,0]$}{
    $\hat{p} \gets G(I^{t})$\Comment*[r]{Predict the next token}
    $I^{t-1} \gets \enc(I^{t}, \hat{p})$\;
}
$\hat{I} \gets I^0$\;
\Return $\hat{I}$\;
\end{algorithm}

Next, in SEEK, the adversary reconstructs the masked image pixel by pixel to remove the embedded watermark.
The vulnerability rank provided by $H(X)$ is now used to guide this reconstruction process, where pixels having lower scores are reconstructed later.
The reason is that we cannot assume that any existing generator is the oracle model that provides perfect predictions in the real world.
Thus, in the auto-regressive process, the generative model accumulates errors due to its intrinsic property, where its previous prediction becomes its next input.
The bias of its prediction therefore reaches its peak in its last prediction.
Hence, reconstructing the most vulnerable pixel in the last step results in a greater prediction bias, which enlarges the modification on the pixel.

\begin{figure*}[t!]
\centering
\includegraphics[width=.9\textwidth]{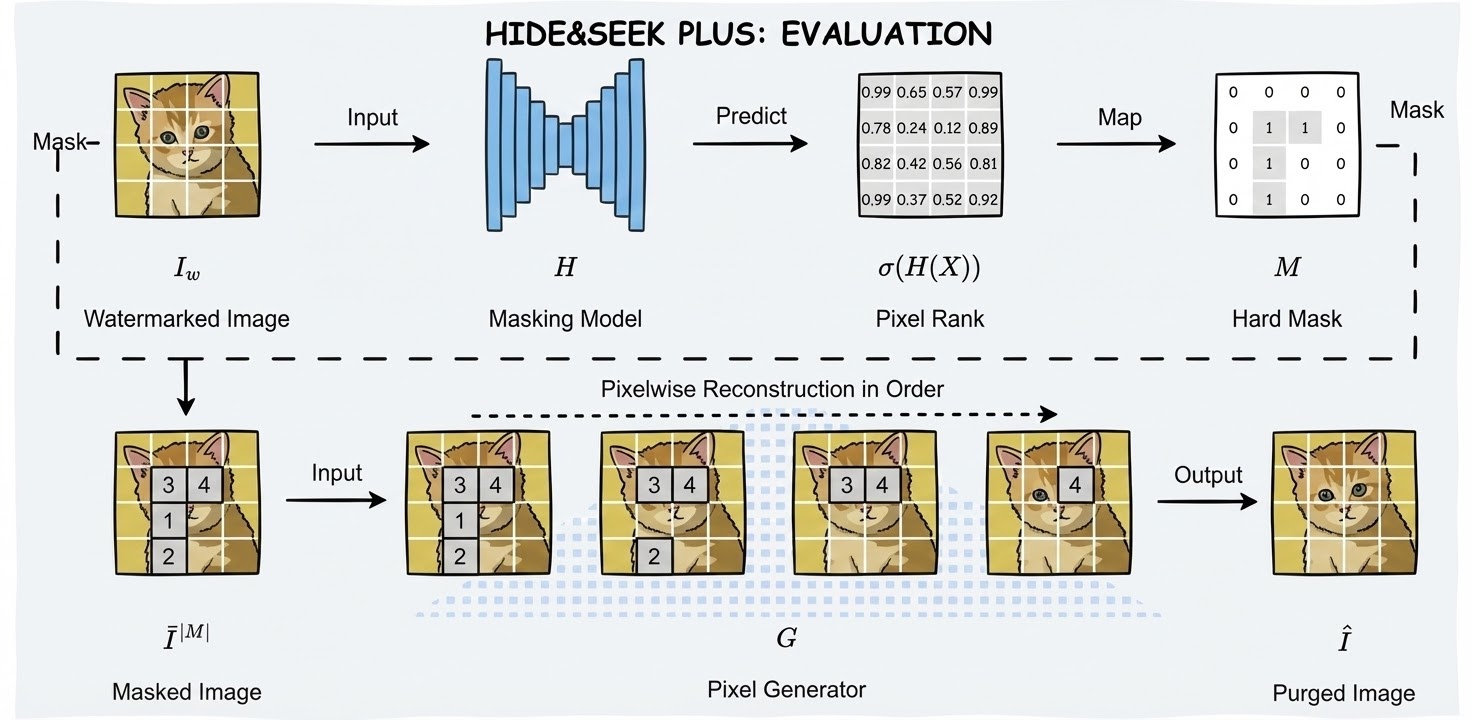}
\caption{\textbf{HIDE\&SEEK Plus Evaluation Workflow}. The watermarked image is fed into the masking model to produce a pixel rank. Pixels with probabilities less than $0.5$ are directly masked for reconstruction. Following the order provided by the pixel rank, the masked pixels are gradually reconstructed, where the most vulnerable pixel is reconstructed at last.}
\label{fig:hsp_workflow_eval}
\end{figure*}

To implement the pixel generator, we employ the state-of-the-art Fractal Generative Network (FGN) \cite{li2025fractal} as the model architecture.
FGN is designed based on the ``Divide \& Conquer" strategy.
This means that $G:\R^{C \times w \times h} \mapsto \R^{C \times S}$ is an auto-regressive model that consists of multiple fractal levels, and $G$ processes a masked image recursively.
Here, $S$ denotes the range pixel values, which is $256$ by default.
At each fractal level, the masked image is divided into multiple patches.
Each patch is then delivered to the next level for further process.
We denote this recursive process as $G(X) = G^n(...G^2(G^1(X)))$, where $G^j$ is the block at the $j$-th fractal layer in $G$.
In each iteration of the pixel-wise reconstruction, $G$ gives logits for only one pixel.
That is, for a mask $M$ of size $|M|$, it takes $|M|$ iterations in total to fully reconstruct the masked image.
Let $\bar{X}^t$ denote the reconstructed image in step $t$, we then have $\bar{X}^{t-1} = \enc(\bar{X}^t, G(\bar{X}^t))$.
$\enc(\cdot, \cdot)$ is an encoding function that takes in $\bar{X}^t$ and the logits of the predicted pixel $G(\bar{X}^t)$ and outputs $\bar{X}^{t-1}$, where the next pixel is reconstructed.

To train $G$, the adversary randomly creates a mask $M \in \{0,1\}^{w \times h}$ to get the masked image $\bar{X} = M \odot X$.
Then, the logit of the next pixel $\hat{p} = G(\bar{X})$ is derived.
The adversary uses the pixel loss defined below to measure the distance between the predicted and the target pixels:
\begin{equation}\label{eq:pix_loss}
\begin{aligned}
    \Lpix(\hat{p}, p) = & - \sum_{c=0}^{C} \sum_{b=0}^{B} \log \frac{e^{\hat{p}_{c, b}}\cdot \onehot(p)_{c, b}}{\sum_{b=0}^{B}e^{\hat{p}_{c,b}}},\\
\end{aligned}
\end{equation}
where $\ppred \in \R^{c \times b}$ denotes the logits predicted by $G$; 
$\ppred_{c,b}$ is the probability that this pixel has value $b$ in channel $c$;
$p$ is the original pixel sharing the same position with $\hat{p}$ in $\tilde{X}$;
$\onehot(\cdot)$ denotes the OneHot encoding function.
This distinguishes FGM from other generators, because an image is no longer treated as an entity but as a sequence of tokens as if in natural language processing.
By minimizing this loss, the adversary trains $G$ to make a correct prediction about the next pixel based on visible pixels.

The adversary also wants to preserve the best perceptual quality.
Hence, the perceptual distance between the reconstructed image $\bar{X}^{t-1}$ and the corresponding target image $X^{t-1}$ at step $t-1$ is measured by the perceptual loss defined as
\begin{equation}\label{eq:pct_loss}
\begin{aligned}
    \Lpct(X^{t-1}, \bar{X}^{t-1}) = \|\alex(X^{t-1}) - \alex(\bar{X}^{t-1}) \|^2,
\end{aligned}
\end{equation}
where $\alex(\cdot)$ denotes the pre-trained Alex Net that extracts the features from the input image;
$X^{t-1}$ denotes the original image with the corresponding pixel revealed.

The complete loss function for optimizing $G$ can then be defined as 
\begin{equation}
\begin{aligned}
L_{\text{SEEK}} = \lambda_4 \Lpix(\hat{p}, p) + \lambda_5 \Lpct(X^{t-1}, \tilde{X}^{t-1})
\end{aligned}
\end{equation}
For each input image $X$, $G$ only takes one step in the training phase for loss computation.
Because the mask is created uniformly randomly, the training is able to cover all states of prediction of the masked pixels.

\prg{Evaluation}
With trained $H$ and $G$, the adversary can launch the watermark removal attack on an arbitrary watermarked image $I$.
As demonstrated in the \textit{evaluation} part of \autoref{alg:hsp}, given $I$, the adversary creates a hard mask with $I$ and $H$, denoted as $M = \text{get\_hard\_mask}(\sigma (H(I)))$.
Here, $\sigma(\cdot)$ is the sigmoid function;
$\text{get\_hard\_mask}(\cdot)$ is a function that maps values in the input greater than $0.5$ to $1$, while those less than $0.5$ to $0$.
Next, the adversary masks $I$ with $M$ to get $\bar{I}^{|M|} = M \odot I$.
Here, we use $|M|$ to denote the total number of steps in the recursive reconstruction process.
Let $t$ be the index of the current step.
In each iteration, $G$ takes in the masked image at step $t$ to predict the next token $\hat{p}=G(\bar{I}^t)$.
Then, the masked image in the next step $\bar{I}^{t-1}$ can be derived with the encoding function, denoted as $\bar{I}^{t-1} = \enc(\bar{I}^{t}, \hat{p})$.
Lastly, with the last masked pixel reconstructed, the step index $t \gets 0$.
The adversary can eventually obtain the purged image $\hat{I} = \bar{I}^0$.

\section{Theoretical Analysis}
In this section, we would like to theoretically discuss how the reconstruction order of the masked pixels affects the attack performance. 
First, we are going to introduce some notations and assumptions. 

\prg{Masking Set}
We denote the collection of the color vectors of the masked pixels by $M$ and the cardinality of $M$ by $|M|$. 
Then, we need $|M|$ steps to recover the image with our pixel generator step by step. 

\prg{Color Vector} 
We denote the original color vector of the pixel we select to predict in the $i$th step by $\bm{x}^{i} = (r_{i},g_{i},b_{i})$ and the prediction in the $i$th step by $\bm{x}_{i} = (\hat{r_{i}},\hat{g_{i}},\hat{b_{i}})$. 

\prg{Element and Order}
The order $(\bm{x}^{1},\bm{x}^{2}, \cdots, \bm{x}^{|M|})$ is determined by the masking model, which means that the perturbation of $\bm{x}^{i}$ may cause fewer discrepancies than the same perturbation of $\bm{x}^{i+1}(i = 1,\dots , |M|-1)$.
In fact, we give a score $\alpha_{i}>0$ for each pixel using the masking model and we estimate the discrepancies by $\alpha_{i}\|\bm{x}_{i} -\bm{x}^{i} \|$ when we change the color vector from $\bm{x}^{i}$ to $\bm{x}_{i}$. 

\prg{Accumulative Prediction Error and Discrepancy}
The accumulative prediction error is given by
\begin{align}
    APE = \left(\sum_{i=1}^{|M|} \|\bm{x}_{i} -\bm{x}^{i} \|^2\right)^{\frac{1}{2}}, 
\end{align}
and the accumulative prediction discrepancy is given by 
\begin{align}
    APD = \sum_{i=1}^{|M|} \alpha_{i}\|\bm{x}_{i} -\bm{x}^{i} \|. 
\end{align}

\prg{Assumptions} 
We assume that the predictor is an oracle predictor such that the error in the $i$th step is independent of the prediction order, and we denote the prediction error in the $i$th step by $y_i>0$. 
For example, if we recover the image by permutation $\pi$ of $\{1, 2, \dots , |M|\}$, then we have $\|x_{\pi(i)} - x^{\pi(i)}\| = y_i(i = 1,2,\dots ,|M|)$.
Then, $APE = \left(\sum_{i=1}^{|M|} y_{i}^2\right)^{\frac{1}{2}}$ is independent of the prediction order. Furthermore, we assume that $y_i\leq y_{i+1}(i = 1,\dots , |M|-1)$. 

\begin{lem}[Rearrangement Inequality, \cite{Day_1972}]\label{Rearrangement Inequality}
    Let $a_1\leq a_2\leq \dots \leq a_n$ and $b_1\leq b_2\leq \dots \leq b_n$ be two sequences of real numbers. Then, for any permutation $\pi$ of $\{1, 2, \dots , n\}$, we have 
    \begin{align}
        \sum_{i = 1}^{n} a_i b_{n-i+1}\leq \sum_{i = 1}^{n} a_i b_{\pi (i)} \leq \sum_{i = 1}^{n} a_i b_i. 
    \end{align}
\end{lem}

\begin{thm}\label{thm:order}
    If we recover the image by ascending order $O_1 = (\bm{x}^{1},\bm{x}^{2}, \cdots, \bm{x}^{|M|})$ and random order $O_2 = (\bm{x}^{\pi (1)},\bm{x}^{\pi (2)}, \cdots, \bm{x}^{\pi (|M|)})$, then \begin{align}\label{APDOI}
        APD (O_2) = \sum_{i=1}^{|M|} \alpha_{i} y_{\pi(i)} \leq APD (O_1) = \sum_{i=1}^{|M|} \alpha_{i} y_{i}. 
    \end{align}
\end{thm}

\begin{proof}
    According to our assumptions, we have $\alpha_i\leq \alpha _{i+1}, y_i\leq y_{i+1}(i = 1,\dots , |M|-1)$. Therefore,  \autoref{APDOI} holds by Lemma \autoref{Rearrangement Inequality}. 
\end{proof}

Theorem \autoref{thm:order} justifies our attack strategy in HS+. 
Theoretically, given a well-trained masking model that provides a correct reconstruction order, a pixel generator is going to achieve the maximum accumulative prediction discrepancy by following the reconstruction order.

\section{Evaluation}
\subsection{Experiment Settings}
The details of image manipulations, datasets, implementations, and baselines for attacks and defenses are listed in \autoref{sect:exp_settings}.

\prg{Evaluation Metrics}
The following metrics are used for evaluating attack performance.
\begin{itemize}
    \item \textbf{Peak Signal-to-Noise Ratio} (PSNR): 
    PSNR is used to measure image quality by comparing an image with its processed version.
    
    \item \textbf{Structural Similarity Index} (SSIM): 
    SSIM is used to measure the perceived quality of an image by comparing a processed version with it.
    
    \item \textbf{Learned Perceptual Image Patch Similarity} (LPIPS): 
    LPIPS measures the similarity between two images by leveraging deep learning to align with human visual perception.
    
    \item \textbf{Bit Accuracy} (BA): Given an original watermark and an extracted watermark in binary form, BA is the ratio of correctly recovered bits to the length of the watermark.
    
    \item \textbf{Inverse Distance} (ID): 
    ID is only used to evaluate the extracted watermark of TRW. 
    A watermark is claimed to be detected only if the $l_1$ distance of the extracted sequence from the watermark is below a pre-defined threshold.
    We follow the previous study \cite{kassis2025unmarker} and set the threshold to $\frac{1}{71}$ to lift the difficulty of the attack task.
    
    \item \textbf{Detection Rate} (DR):
    DR is the ratio of the number of watermarks being detected to the total number of attempts.
\end{itemize}

\subsection{Effectiveness and Fidelity}
\prg{Evaluation of Defense Watermarks}
Before evaluating the attacks, we first evaluate the robustness of the defense watermarks.
The robustness of a defense watermark is measured by the degree of resistance when a watermarked image is attacked.
We test the watermark methods with image manipulations because they are considered easy to implement and efficient to launch.
In this evaluation, we focus on the effectiveness and robustness of the watermarks, as previous studies have demonstrated their covertness.

As illustrated in \autoref{fig:res_manipulation}, the defense watermarks are tested against five types of image manipulations with different parameters.
We exclude metrics for visual quality here because the robustness of defense watermarks is the main concern.
Additionally, these manipulations are not going to drastically affect visual quality of images \cite{kassis2025unmarker,lukas2023ptw}.
Image rotation is not included because this can be simply detected and corrected before the validation process \cite{zhao2024invisible}.

The experimental results demonstrate that the detection rate of each defense watermark reaches $1.00$ when no image manipulation is applied to the watermarked images.
Image cropping is effective in removing GANF watermarks and can decrease the detection rate of HiDDeN and PTW.
However, the rest of the watermarks show great robustness to cropping, which indicates that they are somewhat immune to spatial manipulation.
Spectral manipulations such as blurring, quantization, and JPEG compression can significantly affect the high-frequency components of watermarked images.
Hence, the non-semantic watermarks (from AF to StableSig.) suffer from decreases in detection rate in different levels.
In contrast, semantic watermarks (StegaStamp and TRW) show superior performance, where their detection rate remains $1.00$ in most cases, proving their robustness against image manipulations.

\begin{figure*}[t!]
    \centering
    \includegraphics[width=\textwidth]{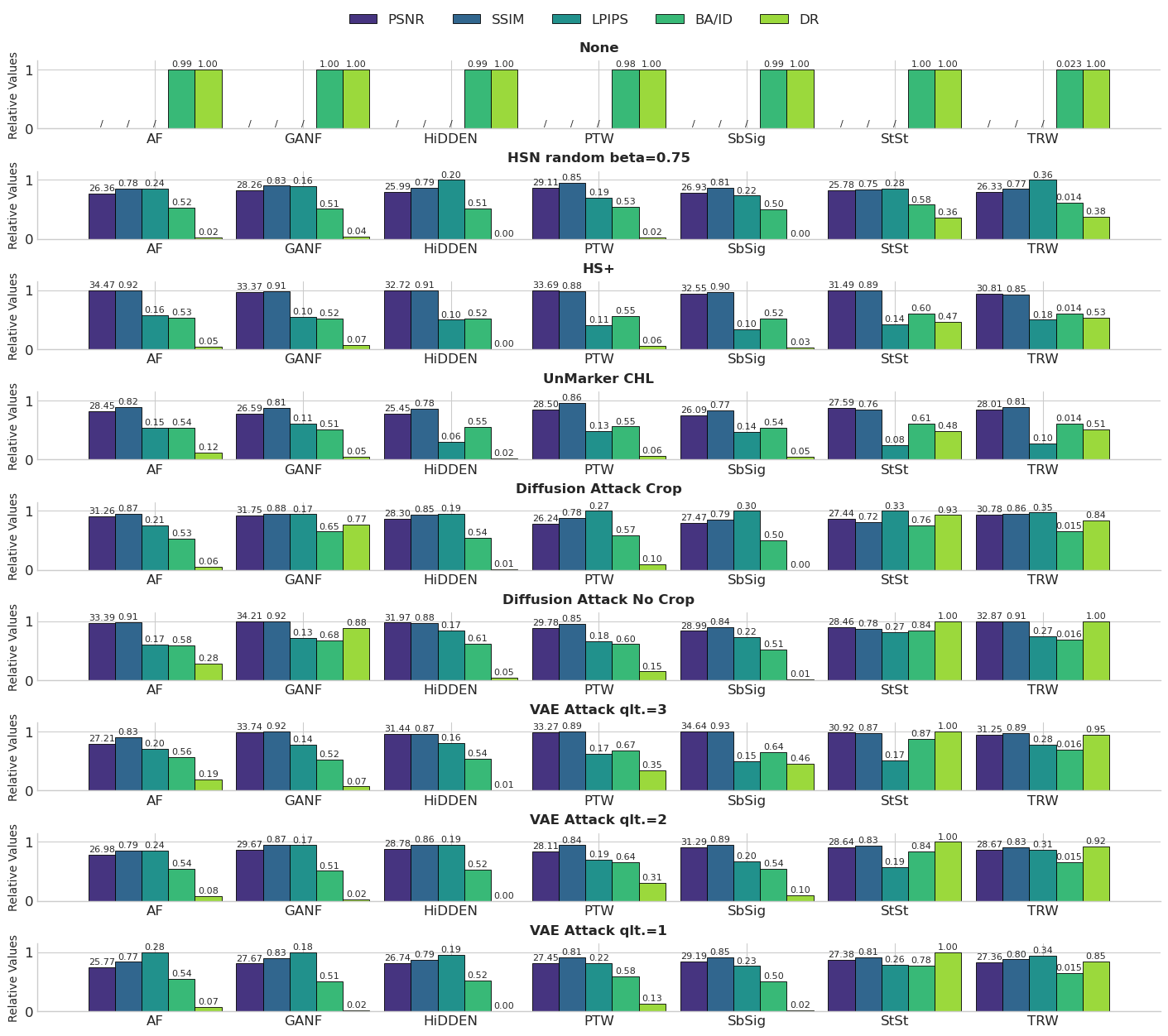}
    \caption{Evaluation of Attacks on Defense Watermarks. Each sub-figure contains the results of a defense watermark against image manipulations. StableSignature and StegaStamp are abbreviated as SbSig and StSt. The fourth bar of TRW represents the inverse distance, whereas the rest of the fourth bars stand for bit accuracy. The height of the bars represents the normalized value under one metric, and the actual values are on top of the bars.}
    \label{fig:eval_attacks}
\end{figure*}

\prg{Evaluation of Watermark Removal Attacks}
Now, we compare HIDE\&SEEK with the SOTA attacks by launching these attacks on the watermarked image generated by the defense watermarks.
As demonstrated in \autoref{fig:eval_attacks}, Diffusion Attack and VAE Attack provide promising results in removing non-semantic watermarks.
However, they fail to remove semantic watermarks because the watermarks residing in the low-frequency components of images are very robust.
The critical idea of such attacks is that the noise introduced must be sufficient to jam the watermark signal.
Next, a powerful generator is employed to recover the purged image.
This explains why Diffusion Attack preserves better visual quality in most cases than VAE attack, because diffusion model is considered more sophisticated than vanilla auto-encoder.
We notice that UnMarker, HSN and HS+ are effective in removing both the semantic and non-semantic watermarks.
This is because these attacks are designed to cause frequency discrepancies in both the low- and high-frequency domains of watermarked images.
Visualized attack results are provided in \autoref{fig:exp_visual}.

HSN shows better attack performance in most cases, where the detection rate is the lowest among all attacks.
When facing GANF, VAE attacks result in a lower detection rate compared to HSN.
This is because GANF is sensitive to spatial manipulations such as cropping, which is shown in \autoref{fig:res_manipulation}.
Both two versions of HIDE\&SEEK do not include cropping or perform global image perturbation to preserve high image quality.
Therefore, it does not always have the best performance compared to the others.
HS+ successfully removes the watermark while preserving high visual quality.
The PSNR values of HS+ remain above $30$ and the SSIM values are all above $0.85$ in all tests .
Although HS+ does not always have the best attack performance, it is effective when faced with HiDDeN and StableSignature, where the detection rate is reduced to $0.0$ and $0.03$.
Semantic watermarks are observed to show impressive robustness when faced with advanced attacks.
Compared to UnMarker, HS+ does not include cropping in the attack process, resulting in an expected higher visual quality.
Further, HS+ outperforms the other attacks in most cases.
This indicates that the perturbation induced by the auto-regressive model in the prediction process suffices to purge the watermark.

\prg{Evaluation of Efficiency}
We further evaluate HS with the SOTA attack in terms of efficiency.
The results are shown in \autoref{tab:eval_efficiency}.
Under the same setting, we randomly select $1,00$ watermarked images to launch the attacks.
We repeat this operation $10$ times to get the average time cost.
VAE attack has the lowest time cost, because the vanilla auto-encoder has a relatively simple structure.
Notably, HSN achieves a result similar to that of VAE attack, but with better attack performance as shown in \autoref{fig:eval_attacks}.
In contrast, Unmarker has the highest time cost because the perturbation is optimized for each watermarked image.
For the convenience of the readers, we convert the time cost into money based on the price of renting one A100 GPU on Amazon Web Services, which is $4.1$\$ per hour.
In conclusion, both HSN and HS+ can provide highly efficient watermark removal attack performance.

\begin{table}[t]
\caption{Average Consumption of the Watermark Removal Attacks on One Hundred Machine-Generated Images}
\label{tab:eval_efficiency}
\centering
\renewcommand{\arraystretch}{.8}
\setlength{\tabcolsep}{4pt}
\begin{tabular}{c|ccccc}
\toprule
   Attacks & HSN & HS+ & UnMarker & Diff Atk & VAE Atk\\
\midrule
    Time(s) & 871 & 1302 & 2553 & 1454 & 839 \\
    Money(\$) & 0.99 & 1.48 & 2.91 & 1.66 & 0.96 \\
\bottomrule
\end{tabular}
\end{table}

\subsection{Ablation Study}
\prginit{Impact of Masking Strategy and Ratio}
We first investigate how the masking strategy and ratio would affect the attack performance of HSN.
Three masking strategies are combined with masking ratios selected from $\{0.60, 0.65, 0.70, 0.75, 0.80\}$ to evaluate their impact on HSN faced with StegaStamp.
StegaStamp is selected because the robustness of StegaStamp allows HSN to have various results, ranging from low to high rather than having only low detection rates.
As demonstrated in \autoref{fig:ablation_masking}, scattered masking preserves the best image quality but results in high detection rates.
Given $\beta=0.60$, scattered HSN can have $31.18$ PSNR and $0.88$ SSIM, whereas the watermark detection rate is $1.0$.
Even with increased $\beta$, scattered HSN is less effective in removing the watermark compared to the other two strategies under the same setting.
This indicates that scattered masking tends to preserve more original frequency information.

Additionally, continuous masking leads to poorer image quality when $\beta$ is low compared to random masking, but they share similar image quality as $\beta$ increases.
This is because a larger $\beta$ covers more area of the watermarked image, leading to two masks that have more overlaps.
We also notice that when $\beta=0.60$, both random and continuous masking have already provided considerable attack performance, where the detection rate is reduced to $0.58$ and $0.55$.
Continuous masking tends to have lower detection rates given a small $\beta$ than random masking, because each of the selected patches is adjacent to another of them.
This induces a more intense perturbation on the frequency domain.

\begin{table}[t]
\caption{Impact of the Loss Terms}
\label{tab:abl_loss}
\centering
\small
\renewcommand{\arraystretch}{1.}
\setlength{\tabcolsep}{4pt}
\begin{tabular}{ccc|ccccc}
\toprule
   $\Lfrq$ & $\Lsmt$ & $\Lpct$ & PSNR$\uparrow$ & SSIM$\uparrow$ & LPIPS$\downarrow$ & Bit Acc.$\downarrow$ & Detect$\downarrow$\\
\midrule
    $\checkmark$ & $\checkmark$ & $\checkmark$ & 31.49 & 0.89 & 0.14 & 0.63 & 0.52 \\
    $\times$ & $\checkmark$ & $\checkmark$ & 30.65 & 0.86 & 0.18 & 0.61 & 0.45 \\
    $\checkmark$ & $\times$ & $\checkmark$ & 30.89 & 0.87 & 0.13 & 0.67 & 0.66 \\
    $\checkmark$ & $\checkmark$ & $\times$ & 31.15 & 0.88 & 0.16 & 0.63 & 0.53 \\
\bottomrule
\end{tabular}
\end{table}

\prg{Impact of Loss Terms}
Next, we evaluate how the loss terms in the training process of HS+ affect the attack performance.
The results are listed in \autoref{tab:abl_loss}, where we exclude the results of $(WH - \|M\|^2)$ and $\Lpix$.
When $(WH - \|M\|^2)$ is removed, the number of patches to mask naturally decreases to $0$, so we omit the test.
$\Lpix$ is the fundamental loss term for training the pixel generator and thus cannot be omitted.
StegaStamp is once again used here as the defense watermark for its outstanding robustness.
We remove the effect of a loss term by setting its $\lambda$ to $0$ in the training process.

\begin{table}[t]
\caption{Impact of the Reconstruction Order}
\label{tab:abl_order}
\centering
\small
\renewcommand{\arraystretch}{1.2}
\setlength{\tabcolsep}{5pt}
\begin{tabular}{c|ccccc}
\toprule
   Order & PSNR$\uparrow$ & SSIM$\uparrow$ & LPIPS$\downarrow$ & Bit Acc.$\downarrow$ & Detect$\downarrow$\\
\midrule
    Original & 31.49 & 0.89 & 0.14 & 0.63 & 0.52 \\
    Inverse & 32.18 & 0.90 & 0.12 & 0.65 & 0.58 \\
    Random & 31.91 & 0.89 & 0.14 & 0.63 & 0.54 \\
\bottomrule
\end{tabular}
\end{table}

$\Lfrq$ contributes to locating pixels that have more impact in the image frequency domain.
When $\Lfrq$ is bypassed, the weight of $\Lsmt$ gains, resulting in a masking model that focuses on masking pixels that have a greater impact on the semantic domain of the image.
Hence, the visual quality of purged images decreases while the attack performance increases.
However, when $\Lsmt$ is bypassed, we do not observe any improvement in attack performance as the visual quality of purged images drops.
Since StegaStamp is a semantic watermark, its watermark signal resides in the low-frequency components.
Removing the semantic loss might have resulted in a masking model that cannot mask pixels that significantly affect the semantic domain of the watermarked image.
The reason is that the high-frequency components of an image are easier to modify than the low-frequency components where energy is concentrated.
Lastly, when $\Lpct$ is bypassed, we observe that there is a slight decrease in the visual quality of the purged images, but the attack performance remains at the same level.
This indicates that $\Lpct$ helps improve the visual quality of the purged images.
Since this loss term is not involved in training the masking model, it hardly has an impact on the attack performance.
Thus, we believe that the masking model plays a more important role in the attack process.

\prg{Impact of Reconstruction Order}
We further investigate the impact of the reconstruction order on the attack performance.
Our original idea is to recover pixels following an order, where the least vulnerable pixel is first reconstructed and the most vulnerable pixel is reconstructed at last.
We conduct two parallel experiments with StegaStamp as the defense watermark, where pixels are reconstructed in inverse and random order.
The results are shown in \autoref{tab:abl_order}, where the original order achieves the best attack performance with $0.63$ bit accuracy and $0.52$ detection rate along with the poorest visual quality.
The inverse reconstruction order results in the best visual quality but with the worst attack performance, whereas the random order slightly deviates from the results of the original order.
This indicates that the reconstruction order of the pixel generator is going to affect the visual quality of the purged image.
Furthermore, this trades off the effectiveness and the fidelity of the watermark removal attack.

\section{Discussion}
\prg{Feasibility}
Both HSN and HS+ do not require any access to watermark detector or its feedback.
The generator and masking model can be trained on an NVIDIA RTX5090 with 32GB VRAM, which is affordable to individuals.
Training datasets such as ImageNet are publicly available.
It is thus feasible for individual adversaries to implement such attacks.

Compared to reconstruction attacks and perturbation attacks, HIDE\&SEEK does not globally perturb the watermarked image.
Comprehensive evaluations show that partial masking and pixel reconstruction have induced sufficient discrepancies that remove the watermark.
Furthermore, unlike perturbation attacks such as UnMarker, HIDE\&SEEK does not perform image manipulations such as cropping that significantly decreases the visual quality of the purged image.
Additionally, once trained, the masking model and generator can be used for general attack purposes.
In contrast, for each watermarked image, UnMarker specifically optimizes two perturbations that aim to distort the high- and low-frequency domains of the watermarked image.
This results in an extra time cost if the number of watermarked images increases.

\prg{Future Work}
A general problem that exists among all defense watermarks for MGIs is that an adversary is always able to trade off between the watermark detection rate and the visual quality of the watermarked image.
In other words, decreasing the quality of a watermarked image will eventually result in a purged image where the embedded watermark can no longer be detected.
Advanced watermark removal attack should guarantee low watermark detection rates while improving the quality of purged images.
Meanwhile, our study reveals that current defense watermarks are vulnerable and can be removed with an acceptable cost in visual quality.
More robust watermarking schemes are to be devised and other types of methods for detecting MGIs should also be taken into consideration.

In terms of generalizability, when trained solely on real-world image datasets such as Imagenet, HS performs well on realistic images. 
But HS does not generalize well to cross-domain datasets such as AnimeDL-2M that comprise unrealistic anime images.
However, it is feasible to combine such unrealistic datasets into the training dataset to further boost the generalizability of HS.

\section{Conclusion}
In this study, we devise a novel image watermark removal attack called HIDE\&SEEK that has two versions respectively named HIDE\&SEEK Naive and HIDE\&SEEK Plus.
HSN utilizes a masked auto-encoder to mask the watermarked image and then recover the masked patches to create the purged image.
HS+ employs a masking model to subtly select vulnerable pixels to mask and uses a pixel generator that recursively reconstructs the masked pixels.
The most vulnerable pixel is reconstructed in the last step to induce the greatest discrepancy.
Through comprehensive evaluations, we show that both HSN and HS+ are effective against the current defense watermark.
HSN outperforms SOTA attacks in terms of attack performance, while HS+ preserves high visual quality with a considerably low watermark detection rate when faced with each defense watermark.
Our study further necessitates research in this domain that prevents misuse of GenAI.

\section*{Ethical Considerations}
In conducting this research, we strictly adhered to a "no-harm" principle. Our methodology was designed to evaluate the robustness of watermarking algorithms in a purely theoretical and isolated context, ensuring that no individuals, organizations, or operational systems were negatively impacted.

\prg{Isolation from Operational Systems}
All experiments described in this paper were performed in a fully offline, local environment. We utilized open-source implementations of watermarking protocols and deployed them on our own hardware. At no point did our attack method interact with, query, or stress-test live commercial APIs or external web services. Consequently, this research caused no service disruption, latency, or financial cost to any model providers or platforms.

\prg{Use of Public, Non-Sensitive Data}
The evaluation was conducted exclusively using standard, publicly available academic datasets. We did not harvest data from private users, nor did we target specific content creators. By relying solely on established benchmarks, we ensured that no personal data was processed, and no individual's privacy or intellectual property rights were infringed upon during the course of this study.

\prg{Safe Security Assessment}
Our work treats the watermarking signal as a mathematical artifact to be analyzed, rather than attacking the infrastructure that hosts it. By confining our "red-teaming" efforts to a sandbox environment, we demonstrate that security vulnerabilities can be identified and documented without risk to the broader digital ecosystem or its stakeholders.

\section*{Open Science}
To promote transparency and reproducibility, we will release artifacts including all the code of all empirical studies employed in this paper. 
Our code will be released after acceptance.
All datasets we use are included in the artifact as they are public accessible.

\section*{Generative AI Usage}
In the preparation of this manuscript, we employed generative AI tools strictly for the purpose of linguistic refinement. Specifically, these tools were used to check for grammatical errors and to improve the flow and readability of the text.

We affirm that:
\begin{itemize}
    \item \textbf{No Content Generation:} The core concepts, experimental design, data analysis, and scientific conclusions are entirely the work of the authors.

    \item \textbf{No Synthetic Text:} No substantial portions of the text were generated by AI. The tools acted solely as a copy-editing assistant.

    \item \textbf{Human Oversight:} All suggestions made by the AI were manually reviewed and verified by the authors, who take full responsibility for the accuracy and integrity of the final publication.
\end{itemize}


\bibliographystyle{ACM-Reference-Format}
\bibliography{ref}

@String{Computing = "Computing" }

@String{Computer = "{IEEE} Computer" }

@String{Springer = "Springer-Verlag" }

@article{zhang2022mask,
  title={How mask matters: Towards theoretical understandings of masked autoencoders},
  author={Zhang, Qi and Wang, Yifei and Wang, Yisen},
  journal={Advances in Neural Information Processing Systems},
  volume={35},
  pages={27127--27139},
  year={2022}
}

@article{shen2025hatebench,
  title={HateBench: Benchmarking Hate Speech Detectors on LLM-Generated Content and Hate Campaigns},
  author={Shen, Xinyue and Wu, Yixin and Qu, Yiting and Backes, Michael and Zannettou, Savvas and Zhang, Yang},
  year={2025},
  publisher={CISPA}
}

@inproceedings{ma2025meme,
  title={From Meme to Threat: On the Hateful Meme Understanding and Induced Hateful Content Generation in Open-Source Vision Language Models},
  author={Ma, Yihan and Shen, Xinyue and Qu, Yiting and Yu, Ning and Backes, Michael and Zannettou, Savvas and Zhang, Yang},
  booktitle={USENIX Security Symposium (USENIX Security). USENIX},
  year={2025}
}

@inproceedings{shen2025gptracker,
  title={GPTracker: A Large-Scale Measurement of Misused GPTs},
  author={Shen, Xinyue and Shen, Yun and Backes, Michael and Zhang, Yang},
  booktitle={2025 IEEE Symposium on Security and Privacy (SP)},
  pages={336--354},
  year={2025},
  organization={IEEE}
}

@article{tang2025towards,
  title={Towards Extensible Detection of AI-Generated Images via Content-Agnostic Adapter-Based Category-Aware Incremental Learning},
  author={Tang, Shuai and He, Peisong and Li, Haoliang and Wang, Wei and Jiang, Xinghao and Zhao, Yao},
  journal={IEEE Transactions on Information Forensics and Security},
  year={2025},
  publisher={IEEE}
}

@article{zhu2023genimage,
  title={Genimage: A million-scale benchmark for detecting ai-generated image},
  author={Zhu, Mingjian and Chen, Hanting and Yan, Qiangyu and Huang, Xudong and Lin, Guanyu and Li, Wei and Tu, Zhijun and Hu, Hailin and Hu, Jie and Wang, Yunhe},
  journal={Advances in Neural Information Processing Systems},
  volume={36},
  pages={77771--77782},
  year={2023}
}

@article{tramer2019adversarial,
  title={Adversarial training and robustness for multiple perturbations},
  author={Tramer, Florian and Boneh, Dan},
  journal={Advances in neural information processing systems},
  volume={32},
  year={2019}
}

@inproceedings{hayes2018learning,
  title={Learning universal adversarial perturbations with generative models},
  author={Hayes, Jamie and Danezis, George},
  booktitle={2018 IEEE Security and Privacy Workshops (SPW)},
  pages={43--49},
  year={2018},
  organization={IEEE}
}

@inproceedings{lee2023robust,
  title={Robust evaluation of diffusion-based adversarial purification},
  author={Lee, Minjong and Kim, Dongwoo},
  booktitle={Proceedings of the IEEE/CVF International Conference on Computer Vision},
  pages={134--144},
  year={2023}
}

@inproceedings{nie2022diffusion,
  title={Diffusion Models for Adversarial Purification},
  author={Nie, Weili and Guo, Brandon and Huang, Yujia and Xiao, Chaowei and Vahdat, Arash and Anandkumar, Animashree},
  booktitle={International Conference on Machine Learning},
  pages={16805--16827},
  year={2022},
  organization={PMLR}
}

@article{Day_1972, title={Rearrangement Inequalities}, volume={24}, DOI={10.4153/CJM-1972-093-x}, number={5}, journal={Canadian Journal of Mathematics}, author={Day, Peter W.}, year={1972}, pages={930–943}}

@inproceedings{wang2024must,
  title={Must: Robust image watermarking for multi-source tracing},
  author={Wang, Guanjie and Ma, Zehua and Liu, Chang and Yang, Xi and Fang, Han and Zhang, Weiming and Yu, Nenghai},
  booktitle={Proceedings of the AAAI Conference on Artificial Intelligence},
  volume={38},
  number={6},
  pages={5364--5371},
  year={2024}
}

@inproceedings{yang2024gaussian,
  title={Gaussian shading: Provable performance-lossless image watermarking for diffusion models},
  author={Yang, Zijin and Zeng, Kai and Chen, Kejiang and Fang, Han and Zhang, Weiming and Yu, Nenghai},
  booktitle={Proceedings of the IEEE/CVF Conference on Computer Vision and Pattern Recognition},
  pages={12162--12171},
  year={2024}
}

@inproceedings{ronneberger2015u,
  title={U-net: Convolutional networks for biomedical image segmentation},
  author={Ronneberger, Olaf and Fischer, Philipp and Brox, Thomas},
  booktitle={Medical image computing and computer-assisted intervention--MICCAI 2015: 18th international conference, Munich, Germany, October 5-9, 2015, proceedings, part III 18},
  pages={234--241},
  year={2015},
  organization={Springer}
}

@inproceedings{jiang2023evading,
  title={Evading watermark based detection of ai-generated content},
  author={Jiang, Zhengyuan and Zhang, Jinghuai and Gong, Neil Zhenqiang},
  booktitle={Proceedings of the 2023 ACM SIGSAC Conference on Computer and Communications Security},
  pages={1168--1181},
  year={2023}
}

@inproceedings{an2024waves,
  title={WAVES: benchmarking the robustness of image watermarks},
  author={An, Bang and Ding, Mucong and Rabbani, Tahseen and Agrawal, Aakriti and Xu, Yuancheng and Deng, Chenghao and Zhu, Sicheng and Mohamed, Abdirisak and Wen, Yuxin and Goldstein, Tom and others},
  booktitle={Proceedings of the 41st International Conference on Machine Learning},
  pages={1456--1492},
  year={2024}
}

@inproceedings{radford2021clip,
  title={Learning transferable visual models from natural language supervision},
  author={Radford, Alec and Kim, Jong Wook and Hallacy, Chris and Ramesh, Aditya and Goh, Gabriel and Agarwal, Sandhini and Sastry, Girish and Askell, Amanda and Mishkin, Pamela and Clark, Jack and others},
  booktitle={International conference on machine learning},
  pages={8748--8763},
  year={2021},
  organization={PmLR}
}

@inproceedings{he2022masked,
  title={Masked autoencoders are scalable vision learners},
  author={He, Kaiming and Chen, Xinlei and Xie, Saining and Li, Yanghao and Doll{\'a}r, Piotr and Girshick, Ross},
  booktitle={Proceedings of the IEEE/CVF conference on computer vision and pattern recognition},
  pages={16000--16009},
  year={2022}
}

@inproceedings{saberi2024robustness,
title={Robustness of {AI}-Image Detectors: Fundamental Limits and Practical Attacks},
author={Mehrdad Saberi and Vinu Sankar Sadasivan and Keivan Rezaei and Aounon Kumar and Atoosa Chegini and Wenxiao Wang and Soheil Feizi},
booktitle={The Twelfth International Conference on Learning Representations},
year={2024},
url={https://openreview.net/forum?id=dLoAdIKENc}
}

@inproceedings{lukas2024leveraging,
title={Leveraging Optimization for Adaptive Attacks on Image Watermarks},
author={Nils Lukas and Abdulrahman Diaa and Lucas Fenaux and Florian Kerschbaum},
booktitle={The Twelfth International Conference on Learning Representations},
year={2024},
url={https://openreview.net/forum?id=O9PArxKLe1}
}

@article{shi2023black,
  title={Black-box backdoor defense via zero-shot image purification},
  author={Shi, Yucheng and Du, Mengnan and Wu, Xuansheng and Guan, Zihan and Sun, Jin and Liu, Ninghao},
  journal={Advances in Neural Information Processing Systems},
  volume={36},
  pages={57336--57366},
  year={2023}
}

@inproceedings{zhou2024dataelixir,
  title={Dataelixir: Purifying poisoned dataset to mitigate backdoor attacks via diffusion models},
  author={Zhou, Jiachen and Lv, Peizhuo and Lan, Yibing and Meng, Guozhu and Chen, Kai and Ma, Hualong},
  booktitle={Proceedings of the AAAI Conference on Artificial Intelligence},
  volume={38},
  number={19},
  pages={21850--21858},
  year={2024}
}

@article{zhao2024invisible,
  title={Invisible image watermarks are provably removable using generative ai},
  author={Zhao, Xuandong and Zhang, Kexun and Su, Zihao and Vasan, Saastha and Grishchenko, Ilya and Kruegel, Christopher and Vigna, Giovanni and Wang, Yu-Xiang and Li, Lei},
  journal={Advances in Neural Information Processing Systems},
  volume={37},
  pages={8643--8672},
  year={2024}
}

@inproceedings{kassis2023breaking,
  title={Breaking security-critical voice authentication},
  author={Kassis, Andre and Hengartner, Urs},
  booktitle={2023 IEEE Symposium on Security and Privacy (SP)},
  pages={951--968},
  year={2023},
  organization={IEEE}
}

@inproceedings{li2020face,
  title={Face x-ray for more general face forgery detection},
  author={Li, Lingzhi and Bao, Jianmin and Zhang, Ting and Yang, Hao and Chen, Dong and Wen, Fang and Guo, Baining},
  booktitle={Proceedings of the IEEE/CVF conference on computer vision and pattern recognition},
  pages={5001--5010},
  year={2020}
}

@article{li2025fractal,
  title={Fractal generative models},
  author={Li, Tianhong and Sun, Qinyi and Fan, Lijie and He, Kaiming},
  journal={arXiv preprint arXiv:2502.17437},
  year={2025}
}

@inproceedings{jois2024pulsar,
  title={Pulsar: Secure steganography for diffusion models},
  author={Jois, Tushar M and Beck, Gabrielle and Kaptchuk, Gabriel},
  booktitle={Proceedings of the 2024 on ACM SIGSAC Conference on Computer and Communications Security},
  pages={4703--4717},
  year={2024}
}

@inproceedings{tancik2020stegastamp,
  title={Stegastamp: Invisible hyperlinks in physical photographs},
  author={Tancik, Matthew and Mildenhall, Ben and Ng, Ren},
  booktitle={Proceedings of the IEEE/CVF conference on computer vision and pattern recognition},
  pages={2117--2126},
  year={2020}
}

@inproceedings{lukas2023ptw,
  title={$\{$PTW$\}$: Pivotal tuning watermarking for $\{$Pre-Trained$\}$ image generators},
  author={Lukas, Nils and Kerschbaum, Florian},
  booktitle={32nd USENIX Security Symposium (USENIX Security 23)},
  pages={2241--2258},
  year={2023}
}

@inproceedings{yu2021artificial,
  title={Artificial fingerprinting for generative models: Rooting deepfake attribution in training data},
  author={Yu, Ning and Skripniuk, Vladislav and Abdelnabi, Sahar and Fritz, Mario},
  booktitle={Proceedings of the IEEE/CVF International conference on computer vision},
  pages={14448--14457},
  year={2021}
}

@inproceedings{yu2022responsible,
  title={Responsible Disclosure of Generative Models Using Scalable Fingerprinting},
  author={Yu, Ning and Skripniuk, Vladislav and Chen, Dingfan and Davis, Larry S and Fritz, Mario},
  booktitle={International Conference on Learning Representations},
  year={2022}
}

@inproceedings{fernandez2023stable,
  title={The stable signature: Rooting watermarks in latent diffusion models},
  author={Fernandez, Pierre and Couairon, Guillaume and J{\'e}gou, Herv{\'e} and Douze, Matthijs and Furon, Teddy},
  booktitle={Proceedings of the IEEE/CVF International Conference on Computer Vision},
  pages={22466--22477},
  year={2023}
}

@inproceedings{zhu2018hidden,
  title={Hidden: Hiding data with deep networks},
  author={Zhu, Jiren and Kaplan, Russell and Johnson, Justin and Fei-Fei, Li},
  booktitle={Proceedings of the European conference on computer vision (ECCV)},
  pages={657--672},
  year={2018}
}

@article{wen2023tree,
  title={Tree-rings watermarks: Invisible fingerprints for diffusion images},
  author={Wen, Yuxin and Kirchenbauer, John and Geiping, Jonas and Goldstein, Tom},
  journal={Advances in Neural Information Processing Systems},
  volume={36},
  pages={58047--58063},
  year={2023}
}

@inproceedings{frank2020leveraging,
  title={Leveraging frequency analysis for deep fake image recognition},
  author={Frank, Joel and Eisenhofer, Thorsten and Sch{\"o}nherr, Lea and Fischer, Asja and Kolossa, Dorothea and Holz, Thorsten},
  booktitle={International conference on machine learning},
  pages={3247--3258},
  year={2020},
  organization={PMLR}
}

@article{liu2023making,
  title={Making DeepFakes more spurious: evading deep face forgery detection via trace removal attack},
  author={Liu, Chi and Chen, Huajie and Zhu, Tianqing and Zhang, Jun and Zhou, Wanlei},
  journal={IEEE Transactions on Dependable and Secure Computing},
  volume={20},
  number={6},
  pages={5182--5196},
  year={2023},
  publisher={IEEE}
}

@inproceedings{zhao2021multi,
  title={Multi-attentional deepfake detection},
  author={Zhao, Hanqing and Zhou, Wenbo and Chen, Dongdong and Wei, Tianyi and Zhang, Weiming and Yu, Nenghai},
  booktitle={Proceedings of the IEEE/CVF conference on computer vision and pattern recognition},
  pages={2185--2194},
  year={2021}
}

@article{cao2025survey,
  title={A survey of ai-generated content (aigc)},
  author={Cao, Yihan and Li, Siyu and Liu, Yixin and Yan, Zhiling and Dai, Yutong and Yu, Philip and Sun, Lichao},
  journal={ACM Computing Surveys},
  volume={57},
  number={5},
  pages={1--38},
  year={2025},
  publisher={ACM New York, NY}
}

@inproceedings{kassis2025unmarker,
  title={UnMarker: A Universal Attack on Defensive Image Watermarking},
  author={Kassis, Andre and Hengartner, Urs},
  booktitle={2025 IEEE Symposium on Security and Privacy (SP)},
  volume={2},
  number={6},
  pages={8},
  year={2025}
}

@article{chen2024high,
  title={High-frequency matters: Attack and defense for image-processing model watermarking},
  author={Chen, Huajie and Zhu, Tianqing and Liu, Chi and Yu, Shui and Zhou, Wanlei},
  journal={IEEE Transactions on Services Computing},
  volume={17},
  number={4},
  pages={1565--1579},
  year={2024},
  publisher={IEEE}
}

\clearpage
\appendix

\section{Experiment Settings}\label{sect:exp_settings}
\prginit{Image Manipulations} Center Cropping, JPEG Compression, Quantization, Gaussian Blurring and Guided Blurring are chosen to be the image manipulations that distort watermarked images to remove embedded watermarks.
The parameters are selected following the previous studies \cite{kassis2025unmarker,lukas2023ptw}.

\begin{figure*}[t]
    \centering
    \includegraphics[width=\textwidth]{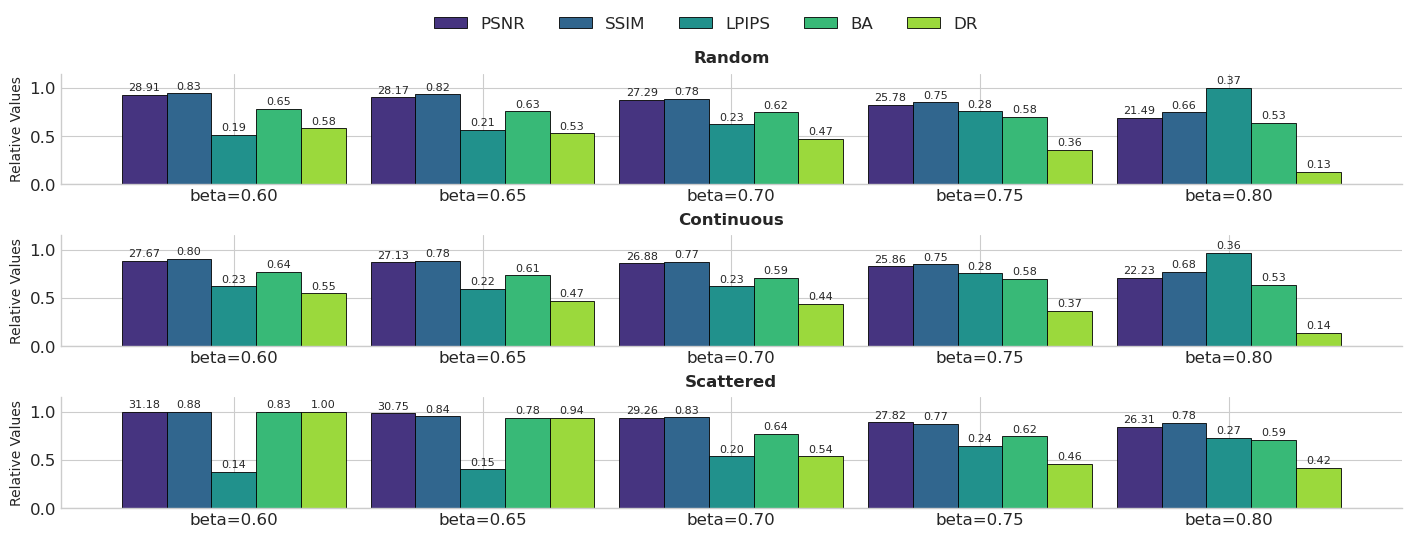}
    \caption{Impact of the Masking Strategy and Ratio. Each sub-figure contains experimental results of a masking strategy with different masking ratio ranging from $0.6$ to $0.8$ with $0.05$ step length. The height of the bars represents the normalized value under one metric, and the actual values are on top of the bars.}
    \label{fig:ablation_masking}
\end{figure*}

\prg{Datasets}
The implemented defense baselines released are trained on various datasets.
Thus, we evaluate the attacks on the datasets that correspond to the defense schemes:
TRW and StableSignature on LAION-5B,
GANF and AF on CelebA,
PTW on FFHQ,
HiDDeN on COCO,
StegaStamp on CelebA-HQ-1024 downscaled to $400 \times 400$.
The evaluation results are the average values derived from $100$ random images following the previous study \cite{jiang2023evading,kassis2025unmarker}.

\prg{Implementation Details}
For the generator in HSN, we employ the Masked Auto-Encoder Vision Transformer \cite{he2022masked} to be its model architecture.
In HS+, the mask model uses the UNet \cite{ronneberger2015u} architecture, and the pixel generator follows the Fractal Generative Model \cite{li2025fractal}.
In evaluation, we directly follow the default settings of the released implementations of the other methods.

\prg{Defense Baselines}
Among all existing image defense watermarking schemes, the following SOTA watermarking methods are selected to be our baselines.
This is because these works fit well into the practical scenario, where the adversary is assumed to have only black-box access to the GenAI APIs and not to have access to the detector.
\begin{itemize}
    \item \textbf{Artificial Fingerprint} (AF) \cite{yu2021artificial} embeds fingerprints in the training dataset that is used to train the generative model.
    The well-trained generative model automatically creates synthetic images with the fingerprints due to transferability.

    \item \textbf{Generative Artificial Network Fingerprinting} (GANF) \cite{yu2022responsible} creates bit-string fingerprint embeddings and incorporates the embeddings into multiple latent outputs of the generative model.
    Synthetic images of the generative model thus carry the fingerprints that can be bit-wise reconstructed for verification.

    \item \textbf{Pivotal Tuning Watermarking} (PTW) \cite{lukas2023ptw} freezes a pre-trained generative model and copies its parameters to create a counterpart.
    The counterpart is then trained with specially designed constraints to have similar performance with the pre-trained model while outputting watermarked images.

    \item \textbf{StegaStamp} \cite{tancik2020stegastamp} embeds invisible bit-string watermarks into synthetic images. 
    The robustness of the watermarks is high and can be successfully reconstructed even if the images are highly distorted in the real world scenario.

    \item \textbf{HiDDeN} \cite{zhu2018hidden} embeds watermarks in synthetic images using an encoder network.
    A noise layer is employed to distort the synthetic image in the training process before the watermarks are reconstructed by a decoder network so as to boost the robustness of the watermark.

    \item \textbf{StableSignature} \cite{fernandez2023stable} first trains a watermark encoder and a decoder to embed and reconstruct bit-string watermarks.
    Next, the decoder of a latent diffusion model is fine-tuned to output watermarked images with a fixed bit-string watermark that can be correctly decoded by the pre-trained decoder.

    \item \textbf{Tree-Ring Watermark} (TRW) \cite{wen2023tree} embeds watermarks in the Fourier space of initial noises in the reverse process of denoising diffusion implicit models.
    After being attacked, the synthetic images are reversed to the initial noises that can be used to recover the watermarks.
    The watermarks are highly resistant to distortion and attacks.
\end{itemize}

\prg{Attack Baselines}
The following SOTA attacks are selected because they have excellent attack performance and fit well into our attack scenario, where the adversary cannot access the detector and has only black-box access to the GenAI APIs.
\begin{itemize}
    \item \textbf{DiffAttack} \cite{saberi2024robustness,an2024waves,zhao2024invisible} employs a diffusion model that recovers a perturbed watermarked image to create a purged image.

    \item \textbf{VAEAttack} \cite{an2024waves,zhao2024invisible} is similar to DiffAttack, but its generative model is changed to a VAE model.

    \item \textbf{UnMarker} \cite{kassis2025unmarker} separately optimizes two perturbations that perturb the high- and low-frequency components of a watermarked image.
\end{itemize}

\begin{figure*}[t]
\centering
\includegraphics[width=\textwidth]{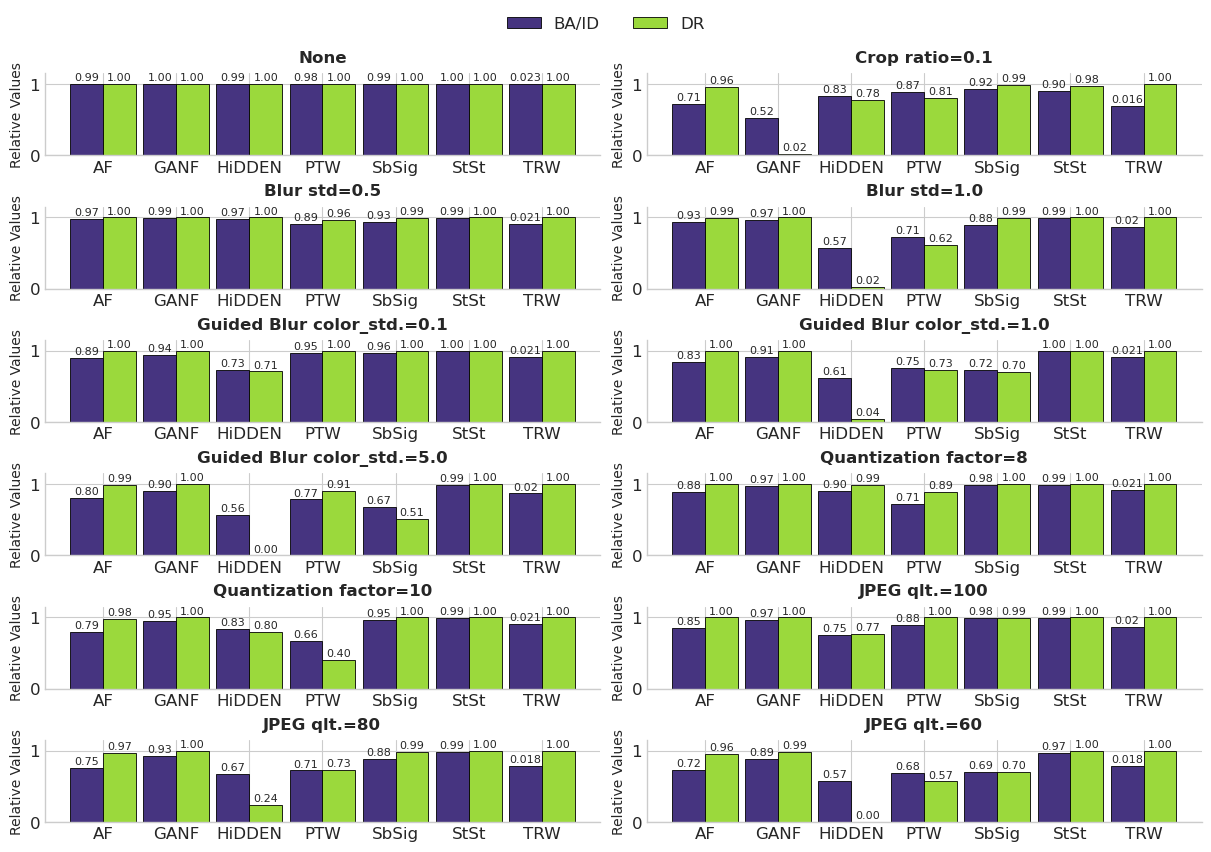}
\caption{Evaluation of Defense Watermarks against Image Manipulations. Each sub-figure contains the results of an image manipulation attacking different defense watermarks. StableSignature and StegaStamp are abbreviated as SbSig and StSt. The blue bar of TRW represents the inverse distance, whereas the rest of the blue bars stand for bit accuracy. The height of the bars represents the normalized value under one metric, and the actual values are on top of the bars.}
\label{fig:res_manipulation}
\end{figure*}

\begin{figure*}[t]
\centering
\includegraphics[width=\textwidth]{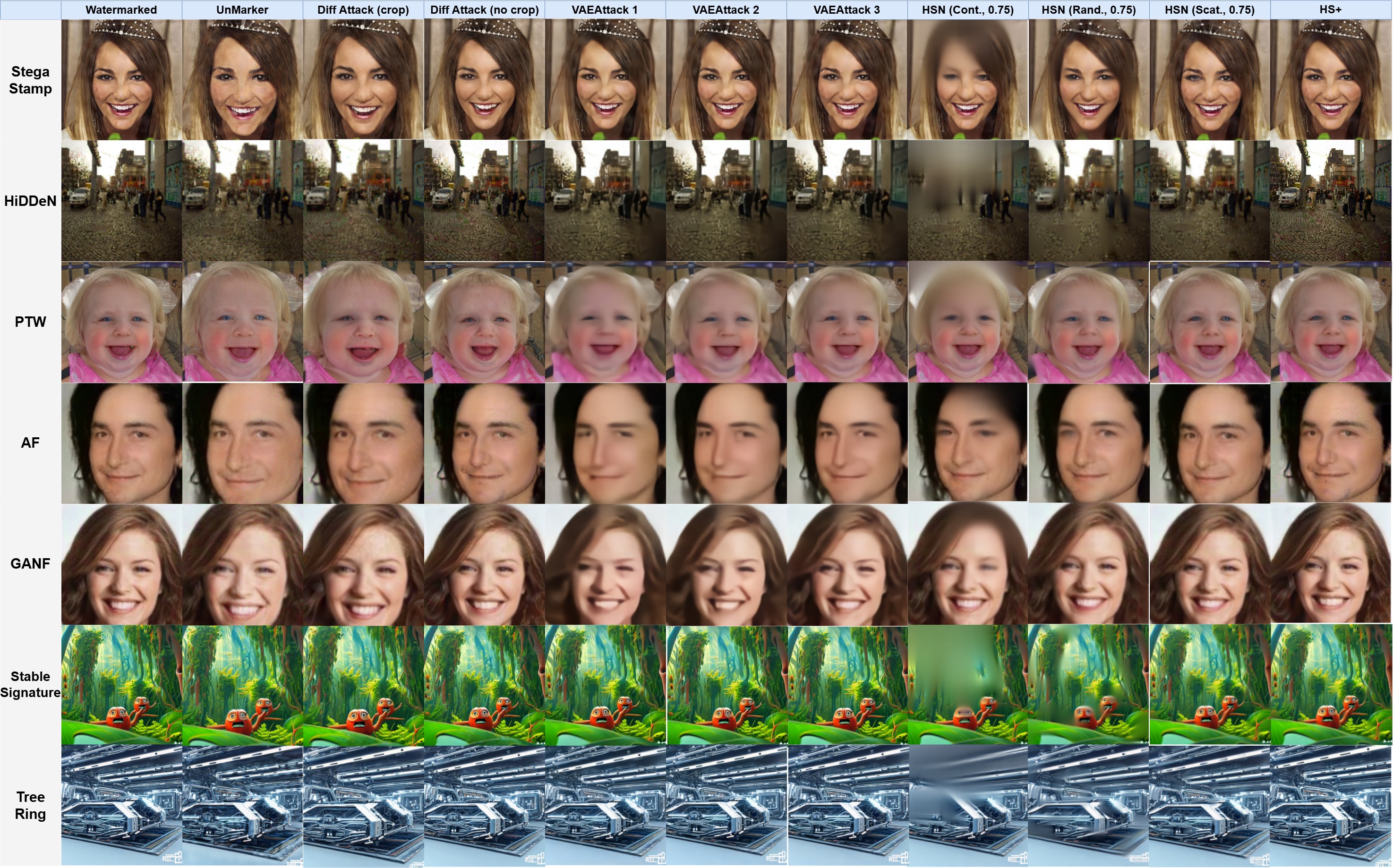}
\caption{\textbf{Visual Effect of Watermarked Images Attacked by Different Methods.}}
\label{fig:exp_visual}
\end{figure*}

\end{document}